\documentclass[bj,authoryear]{imsart}

\usepackage{float}
\usepackage{accents}
\usepackage{stackrel}
\usepackage{graphicx}
\usepackage{subcaption}

\startlocaldefs
\theoremstyle{plain}

\newtheorem{theorem}{Theorem}[section]

\theoremstyle{remark}

\newtheorem{proposition}{Proposition}[section]

\endlocaldefs
\newcommand{\bu}{{\mathbf{u}}}
\newcommand{\by}{{\mathbf{y}}}
\newcommand{\cv}{\mbox{CV}}

\newcommand{\e}{{\mathbf{e}}}

\begin{document}

\begin{frontmatter}
\title{On cross-validated estimation of skew normal model}
\runtitle{Cross-validated estimation}

\begin{aug}
\author[A]{\inits{F.}\fnms{JIAN}~\snm{ZHANG}\ead[label=e1]{J.Zhang-79@kent.ac.uk}\orcid{0000-0001-8405-2323}}
\author[B]{\inits{S.}\fnms{TONG}~\snm{WANG}\ead[label=e2]{ianwang1009@gmail.com}}
\address[A]{School of Mathematics, Statistics and Actuarial Science,
University of Kent, Canterbury, Kent CT2 7NF, United Kingdom\printead[presep={,\ }]{e1}}

\address[B]{Novatis Pharmaceuticals UK, The WestWorks, London W12 7FQ, United Kingdom \printead[presep={,\ }]{e2}}
\end{aug}

\begin{abstract}
Skew normal model suffers from inferential drawbacks, namely singular Fisher information in the vicinity of symmetry and diverging of maximum likelihood estimation. To address the above drawbacks, Azzalini and Arellano-Valle (2013) introduced maximum penalised likelihood estimation (MPLE) by subtracting a penalty function from the log-likelihood function with a pre-specified penalty coefficient. Here, we propose a cross-validated MPLE to improve its performance when the underlying model is close to symmetry. We develop a theory for MPLE, where an asymptotic rate for the cross-validated penalty coefficient is derived. We further show that the proposed cross-validated MPLE is asymptotically efficient under certain conditions.  In simulation studies and a real data application, we demonstrate that the proposed estimator can outperform the conventional MPLE when the model is close to symmetry.
\end{abstract}


\end{frontmatter}

\section{Introduction}
Skewness, which measures the asymmetry of a distribution, is an important data feature to characterise. Change of data skewness can serve as a basis for detecting an attack upon a sensor network, for providing an early warning for abrupt climate changes, for estimating  aggregates of small domain business, for modelling equity excess returns, for characterising sensitivity of anti-cancer drugs, among others \citep{r6, r12, r8, r10, r9}. For example, in cancer research, people are interested in charaterising drug sensitivity and development of novel therapeutics.
 The data considered in this study consist of the measurements of median inhibition concentrations, IC50s, of $227$ drugs in 111 cancer cell lines  \citep{r14}. IC50 is a measure of how much drug is needed to inhibit the multiplication of that cell line by 50\%. The log-IC50 informs the drug sensitivity against  cancer cells. The location, dispersion and skewness parameters of the log-IC50 can be used to search for a combined drug therapy. For example, histogram plots for log-IC50s of drugs Erlotinib and Paclitaxel in Figure \ref{ErPa} demonstrate that these two drugs have contrasting data features, one has positive drug response and the other has negative response (i.e., drug resistance). Erlotinib is an inhibitor of the epidermal growth factor receptor (EGFR) tyrosine kinase pathway while Paclitaxel is a chemotherapy drug. Cancer stem cells are often enriched after chemotherapy and
induce tumor recurrence, which poses a significant clinical
challenge. Combining Erlotinib with Paclitaxel can overcome paclitaxel-resistant cervical cancer \citep{r16}.

By introducing a shape parameter $\alpha$ in a normal distribution, the skew-normal and more generally, skew symmetry distributions, can provide a better fit for asymmetric data than does the normal (\citep{r1}). Because of their appealing mathematical properties and usefulness in practice, skew-normal distributions have received considerable attention in the past two decades. Extensions have been made in various directions, including multivariate skew-normal, skew $t$- and skew elliptical distributions and finite mixtures of skew normals \citep{r2, r3, r15, r4, r17} and references therein.
 Some of these developments, however, suffer from inferential drawbacks, namely singular Fisher information in the vicinity of symmetry and diverging of maximum likelihood estimator of $\alpha$ \citep{r1, r13}.  To eliminate the singularity, a tentative remedy called centred parametrisation (CP) was put forward by \citep{r1}, where  $\alpha$ is reparametrised to Pearson's skewness index. Unfortunately, the likelihood function under the CP is not explicitly available as the Jacobian factor of the transformation  
 is unbounded when the underlying model is symmetric.  
 \cite{r7} showed that even using reparametrisation, the resulting maximum likelihood estimator (MLE) of $\alpha$ has a rate of $n^{-1/6}$ lower than the usual rate of root-$n$ at $\alpha=0$. The above remedy never really caught upon, partly because the mechanism of skewness is unknown in practice and the resulting skew-normal family, under the new parametrisation, loses much of its simplicity \citep{r4}. To develop an alternative remedy to handle the above inferential drawbacks, \cite{r5} considered a penalised log-likelihood by subtracting $\lambda Q(\alpha)$ from the log-likelihood, where penalty $Q(\alpha)$ is used to control the magnitude of $\alpha$. By using Firth's bias-correction technique \citep{r11}, \cite{r4} chose a fixed penalty coefficient $\lambda\approx 0.87591$ for $Q(\alpha)=\log(1+0.85625\alpha^2)$. However, this choice is against our intuition that when the underlying value of $\alpha$ is close to zero, the penalty coefficient should increase to infinity to force the estimating value quickly to zero.  In this paper, we aim to develop a data-driven procedure for improving the choice of the penalty coefficient with a theoretical guarantee.

 In literature, there are two approaches for determining the penalty coefficient, one is information criterion and the other is cross-validation. However, the former, usually working for non-degenerate models, may be invalid for the singular models in which the penalised likelihood (or posterior) cannot be approximated by any normal distribution. Although in a singular model, the generalisation error of an inference procedure may not be estimated well by information criteria, it can be estimated by the cross-validation \citep{r18}. This motivates us to investigate a multifold cross-validation procedure for skew normal estimation. Our contributions to the research field are two-fold.  Firstly, we develop a hyperbolic parametrisation to understand the nature of $\alpha$ from a point view of active function. Under the new parametrisation, we show that  the cross-validated estimator asymptotically attains the Cramer-Rao lower bound to estimating error. By simulation studies and a real data application, we demonstrate that the proposed estimation can outperform the bias-correction approach used in the software ${\bf SN}$ in terms of bias and standard error. Secondly, we unveil a super efficiency for the cross-validated estimation at $\alpha=0$, namely after an appropriate tuning, the cross-validated estimator can recover the true $\alpha$ exactly with an probability tending to one. This is in striking contrast to the MLE which has a very slow convergence rate at $\alpha=0.$ By theoretical analysis and simulations, we show that Firth's bias-correction technique may under-regularise the estimation at $\alpha=0.$ Furthermore, we demonstrate a screening effect of penalisation by simulations that any small skewness will be filtered out by MPLE.

{\centerline {[Put Figure \ref{ErPa} here.]}}
\begin{center}
\begin{figure}[htp]
\includegraphics[height=2.0in,width=0.45\textwidth]{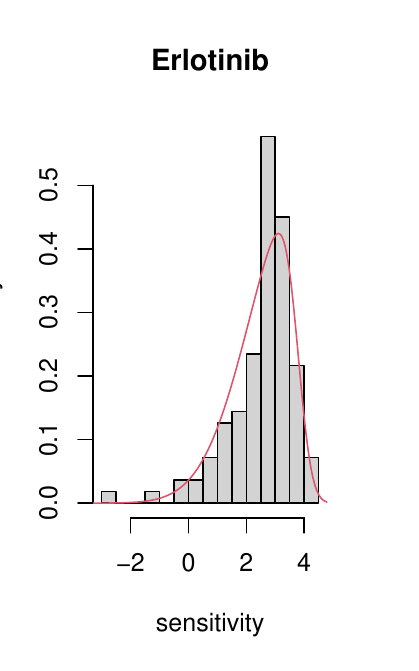}\hfill
\includegraphics[height=2.0in,width=0.45\textwidth]{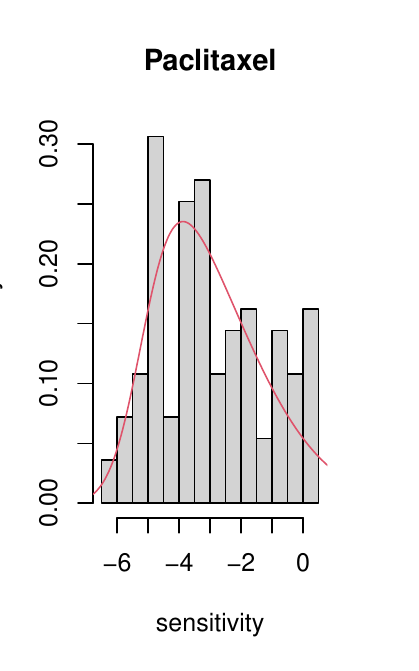}
\caption{\small
{\noindent Histogram plots and skew normal fits of log-IC50 data for drugs Erlotinib and Paclitaxel.
}
}
\label{ErPa}
\end{figure}
\end{center} 

The rest of the article is organised as follows. In Section 2 we develop the cross-validation procedure for skew normal models. We establish theoretical properties
of the proposed procedure in Section 3. We develop a penalised Expectation-Maximisation (EM) algorithm in Section 4. We conduct simulation studies and a real data analysis
in Section 5. We conclude with a discussion in Section 6.  The proofs are relegated to the
Appendix and some simulation results are displayed in the Online Supplementary Material.

\section{Methodology}
In this section we first review the location-scale model for skew normals and reparametrisation. Then after a short discussion of its inherent inferential issues we develop
a multifold cross-validation procedure for the penalised likelihood inference.
\subsection{Location-scale model}
Let $\textrm{SN}(0,1,\alpha)$ denote the standard skew normal distribution with density $f(x;\alpha)=2\phi(x)\Phi(\alpha x),$ $ x\in\mathbb{R},$
where $\phi(x)$ and $\Phi(x)$ are the standard normal density and cumulative distribution
function respectively, and $\alpha$ is the shape parameter to regulate the skewness. The skew normal distribution family includes the standard normal as a special when $\alpha=0$. The skew normal random variable $X$ can be expressed as a linear combination
of two independent variables, a half standard normal $X_{+}$ with density $2\phi(x), x\ge 0$ and the standard normal $X_{0}$,
in the form $X=\delta X_{+}+\sqrt{1-\delta^{2}}X_{0},$ where $\delta=\alpha/\sqrt{1+\alpha^{2}}.$ 
 Consider the location-scale model $Y=\mu+\sigma(X-\delta\sqrt{2/\pi})$ with location parameter $\mu\in {\mathbb R},$ scale parameter $\sigma\ge 0$ and density $f(y;\mu,\sigma,\alpha)=\sigma^{-1}f((y-\mu)/\sigma+\delta\sqrt{2/\pi};\alpha).$
We have $E[Y]=\mu$, $\textrm{var}(Y)=\sigma^{2}\textrm{var}(X)=\sigma^{2}(1-2\delta^{2}/\pi)$, $E[X]=\delta\sqrt{2/\pi}$, and $\textrm{var}(X)=1-2\delta^{2}/\pi.$
  Under the above location-scale model, Pearson's skewness index $\gamma_1$ is related to parameters $\alpha$ and $\delta$ via
\begin{eqnarray}\label{CP}
\gamma_{1}&=&\frac{E\left[(Y-\mu)^{3}\right]}{\textrm{var}(Y)^{3/2}}=\frac{4-\pi}{2}\frac{\delta^{3}(2/\pi)^{3/2}}{(1-2\delta^{2}/\pi)^{3/2}},\quad \left|\gamma_{1}\right|\leq 0.9952. \nonumber\\
\delta&=&\sqrt{\frac{\pi}{2}}\frac{(2\gamma_{1}/(4-\pi))^{1/3}}{\sqrt{1+(2\gamma_{1}/(4-\pi))^{2/3}}},\quad \alpha=\delta/\sqrt{1-\delta^2}.
\end{eqnarray}
Here, $\alpha$ and $\delta$ can be viewed as activation functions of $\gamma_1$ with derivatives
\begin{eqnarray*}
\frac{d\delta}{d\gamma_1}&=&\sqrt{\frac{\pi}2}\frac{2}{3(4-\pi)}\left(\frac{2\gamma_1}{4-\pi}\right)^{-2/3}\left(1+\left( \frac{2\gamma_1}{4-\pi}\right)^{2/3} \right)^{-3/2},\\
\frac{d\alpha}{d\gamma_1}&=&\sqrt{\frac{\pi}2}\frac{2}{3(4-\pi)}\left(\frac{2\gamma_1}{4-\pi}\right)^{-2/3}\left(1+(1-\pi/2)\left( \frac{2\gamma_1}{4-\pi}\right)^{2/3} \right)^{-3/2}
\end{eqnarray*}
which are unbounded derivatives at $\gamma_1=0.$ Figure \ref{activation} demonstrates that the skewness is introduced into the model via the activation function $\delta$ of the input $\gamma_1$, where $\gamma_1$ is restricted to the interval $(-0.9952,0.9952)$. By the activation function, the negative (positive) $\gamma_1$ will be mapped onto strongly negative (positive) $\delta$ while zero $\gamma_1$ will be mapped onto zero $\delta$.
In the CP, Azzalini (1985) used $\gamma_1$ to reparametrise $\alpha$ and $\delta$ through equation (\ref{CP}). 

\subsection{Hyperbolic reparamaterisation}

To tackle the above issue of unboundedness, we reparametrise $\alpha$ by the scaled inverse hyperbolic transformation
$
\theta=\mbox{arcsinh}(\alpha)/a, 
$
where $a>0$ is a pre-selected constant.
This gives rise to a family of activation functions
\[
\alpha=\mbox{sinh}(a\theta)=\frac{1}{2}(e^{a\theta}-e^{-a\theta}),\qquad \delta=\mbox{tanh}(a\theta)=\frac{e^{2a\theta}-1}{e^{2a\theta}+1},\qquad \alpha=\frac{\delta}{\sqrt{1-\delta^2}}
\]
with derivatives
\[
\frac{d\alpha}{d\theta}=a\mbox{cosh}(a\theta)=\frac a2(e^{a\theta}+e^{-a\theta}),\qquad
\frac{d\delta}{d\theta}=a(1-\delta^2).
\]
When $a=1$,  $\theta$ reduces to the Fisher transformation of $\delta.$
Figure 1 in the Online Supplementary Material shows that as $a$ tends to infinity $\mbox{sinh}(a\theta)$ is close to the CP. However, for finite fixed $a$'s, these activation functions give bounded derivatives. In the following, we focus on the simple case $a=1$.
To remove the constraint $\sigma\ge 0$, we reparametrise $\sigma$ by $\eta=\log(\sigma)$, which has a range of $(-\infty,\infty)$ and $\frac{d\sigma}{d\eta}=\sigma$.

\subsection{Initial estimation}
 Given an i.i.d. sample $\mathbf{y}=(y_{i})_{1\le i\le n}$ drawn from the above location-scale model, we use the method of moments to construct initial estimators as follows. 
Setting the first three moments of $Y$ equal to their corresponding sample moments and using the relationships between $\gamma_1$, $\delta$, $\alpha$ and $\theta$, we have the  initial estimates $\mu^{(0)}, \sigma^{(0)}, \eta^{(0)},$ $\theta^{(0)},\alpha^{(0)}$ and $\delta^{(0)}.$
Let $\mu_0, \sigma_0,\eta_0, \theta_0,\alpha_0$ and $\delta_0$ be the ground-truth of the parameters in the model.
It follows from the central limit theory that $\mu^{(0)}=\mu_0+O_p(1/\sqrt{n})$ and
$\gamma_{1}^{(0)}=\gamma_0+O_p(1/\sqrt{n}).$ Using equation (\ref{CP}), for $\theta_0=0$,
we have
$ \delta^{(0)}=O_p(n^{-1/6}),$ $\theta^{(0)}=O_p(n^{-1/6}),$ 
$\alpha^{(0)}=O_p(n^{-1/6}),$ and $ \eta^{(0)}=\eta_0+O_p(n^{-1/3})$.
In contrast, for $\theta_0\not=0,$ we have the following standard root-$\sqrt{n}$ convergence rates,
$
 \delta^{(0)}=\delta_0+O_p(1/\sqrt{n}),$ $ \theta^{(0)}=\theta_0+O_p(1/\sqrt{n}),$
 $\alpha^{(0)}=\alpha_0+O_p(1/\sqrt{n}),$ and
$\eta^{(0)}=\eta_0+O_p(1/\sqrt{n}).$

\subsection{Maximum penalised likelihood estimation}

 To develop an asymptotic theory for maximum penalised likelihood estimation, using the above initial estimates, we can restrict the domain of unknown parameters to
\[
\Omega_n=\{(\mu,\eta,\theta): \mid \mu-\mu^{(0)}\mid \le c_{\mu_0}n^{-1/2}, \mid \eta-\eta^{(0)}\mid \le
c_{\eta_0}n^{-1/3}, \mid \theta-\theta^{(0)}\mid \le c_{\theta_0}n^{-1/6}\}
\]
 for some positive constants $c_{\mu_0}$, $c_{\eta_0}$ and $c_{\theta_0}.$ Note that $\Omega_n$ is shrinking to the true values of the parameters as $n$ tends to infinity.
For the simplicity of notation, let $z_i=(y_i-\mu)/\sigma.$
Given the sample ${\mathbf y},$ we have the following log-likelihood
\begin{eqnarray*}
l_{inc}=l_{inc}\left(\mu,\eta,\theta\mid\mathbf{y}\right) 
 &=&  \frac{n}{2}\log\left(\frac{2}{\pi\sigma^{2}}\right)-\frac{1}{2}\sum_{i=1}^{n}\left(z_i+\delta\sqrt{{2}/{\pi}}\right)^{2}
  +\sum_{i=1}^{n}\log\Phi\left(\alpha\cdot\left(z_i+\delta\sqrt{{2}/{\pi}}\right)\right),
\end{eqnarray*}
and define the maximum likelihood estimate (MLE)
$$
(\hat{\mu},\hat{\eta},\hat{\theta})  = \underset{\left(\mu,\eta,\theta\right)\in \Omega_n}{\arg\max }\ l_{inc}\left(\mu,\eta,\theta\mid\mathbf{y}\right).
$$
The above likelihood is singular at $\theta=0$ with a stationary point at $\theta=0$ regardless values of the other parameters as the Fisher information matrix at
$(\mu,\eta,0)$,
\[
-E\left(\frac{\partial^2\log l_{inc}}{\partial (\mu,\eta,\theta)\partial (\mu,\eta,\theta)^T}\right) |_{(\mu,\eta,0)}=\mbox{diag}\left(\frac{n}{\sigma^{2}},2n,0\right)
\]
is degenerate. This results in a slow convergence rate of $\hat{\theta}$ and non-standard asymptotic behavior of the MLE when the underlying vaule of $\theta$ is zero or near zero.
As noted previously, the MLE of the shape parameter of the skew normal diverges with a probability that is non-negligible for small and moderate sample sizes.
 To address these issues, following \cite {r5}, we maximise the penalised log-likelihood
$$l_{incp}(\mu,\eta,\theta\mid{\mathbf y})=l_{inc}(\mu,\eta,\theta\mid{\mathbf y})-\lambda \mbox{ pen}(\theta),$$
where a penalty is used to control the size of $\theta$, satisfying the condition
\[
\mbox{\it C1 }: \mbox{ pen}(\theta)\ge 0,\quad \mbox{ pen}(0)=\mbox{ pen}'(0)=0,\quad \mbox{ pen}''(0)=2,\quad \lim_{|\theta|\to \infty}\mbox{ pen}(\theta)\to\infty.
\]
For example, hyperbolic penalty $\mbox{pen}_1(\theta)=\alpha^2=(e^{\theta}-e^{-\theta})^2/4$, ridge penalty $\mbox{pen}_2(\theta)=\theta^2$ and log-Cauchy $\mbox{pen}_3(\theta)=\log(1+c_2\alpha^2)=\log(1+c_2(e^{\theta}-e^{-\theta})^2/4)$ meet these conditions, where the log-Cauchy was proposed by \cite{r4}. All these penalties encourage shrinkage of the skewness parameter toward zero while preventing it from diverging to infinity. 
As the Fisher information matrix of the penalised likelihood is not singular,
the following inference can be conducted.
For each $0\le \lambda/n\le \omega_0$, define the maximum penalised likelihood estimator 
(MPLE) 
$$
(\hat{\mu}_{\lambda},\hat{\eta}_{\lambda},\hat{\theta}_{\lambda})
=\mbox{arg}\max_{ (\mu,\eta,\theta)\in\Omega_n} l_{incp}(\mu,\eta,\theta\mid{\mathbf y}).
$$
 The larger the penalty coefficient, the greater the accuracy 
of estimating $\theta$ when the true value of $\theta$ is zero while the larger estimating bias when the true value of $\theta$ is not zero. Multifold cross-validation below strikes a balance between the accuracy and the bias by tuning the penalty coefficient and therefore achieves a better prediction for new samples.

\subsection{Multifold cross validation}
Given an i.i.d. sample ${\mathbf y}$ of size $n$ drawn from a skew normal distribution, for a pre-specified positive constant $\omega_0$, define the expected out-of-sample generalisation error, 
\[
\cv(\lambda)=-E[l_{inc}(\hat{\mu}_{\lambda},\hat{\eta}_{\lambda},\hat{\theta}_{\lambda}|{\mathbf y}^*)],
\]
 if we were to apply the model based on estimator $(\hat{\mu}_{\lambda},\hat{\eta}_{\lambda},\hat{\theta}_{\lambda})$ to predict a new set of observations ${\mathbf y}^*$ drawn independently from the same distribution as that of ${\mathbf y}$.
 The above expectation is taken with respect to both ${\mathbf y}$ and ${\mathbf y}^*$. The generalisation error $\cv(\lambda)$ can be used as a criterion to compare candidate estimators $(\hat{\mu}_{\lambda},\hat{\eta}_{\lambda},\hat{\theta}_{\lambda})$, $0\le\lambda/n\le \omega_0$. We estimate the expected out-of-sample generalisation error $\cv(\lambda)$ by multifold cross-validation as follows. 

For a pre-specified integer $K>0$, divide the data $\mathbf{y}$ into $K$ groups ${\mathbf y}_{j}, 1\le j\le K$ with corresponding index groups $[j], 1\le j\le K$.
For each $\lambda$ and $1\le j\le K$, we calculate estimates $\left(\hat{\mu}_{[-j]\lambda},\hat{\eta}_{[-j]\lambda},\hat{\theta}_{[-j]\lambda}\right)$, based on the training set ${\bf y}_{[-j]}$,  by maximising log-likelihood $l_{incp}\left(\mu,\eta,\theta\mid {\mathbf y}_{[-j]}\right)$. For each $j,$ taking
${\mathbf y}_{j}$ as the validation sample, we can estimate the out-of-sample generalisation error by
$-l_{inc}\left(\hat{\mu}_{[-j]\lambda},\hat{\eta}_{[-j]\lambda},\hat{\theta}_{[-j]\lambda}\mid {\mathbf y}_{j}\right).$ Averaging these estimated errors, we have the following avarage generalisation error
\[
\cv_a(\lambda)=-\frac 1{K}\sum_{j=1}^{K}l_{inc}\left(\hat{\mu}_{[-j]\lambda},\hat{\eta}_{[-j]\lambda},\hat{\theta}_{[-j]\lambda}\mid {\mathbf y}_{j}\right),
\]
 where $K$ is a positive integer. 
The optimal tuning $\lambda_{op}=\mbox{arg}\min_{0\le\lambda/n\le\omega_0} E[\cv_a(\lambda) ]$ is estimated by 
$$\lambda_{cv}=\mbox{arg}\min_{0\le\lambda/n\le\omega_0}\cv_a(\lambda).$$

\section{Asymptotic theory}

Note that for each $\lambda$, the penalised MLE is obtained by solving simultaneous equations
\begin{eqnarray}\label{mleq}
&&\frac 1n\frac{\partial l_{inc}(\hat{\mu},\hat{\eta},\hat{\theta})\mid \mathbf{y})}{\partial\mu}=0, \quad
\frac 1n\frac{\partial l_{inc}(\hat{\mu},\hat{\eta},\hat{\theta})\mid \mathbf{y})}{\partial\eta}=0,\nonumber \\
&&\frac 1n\frac{\partial l_{inc}(\hat{\mu},\hat{\eta},\hat{\theta})\mid \mathbf{y})}{\partial\theta}-\frac{\lambda}{2n}(e^{2\hat{\theta}}-e^{-2\hat{\theta}
})=0,
\end{eqnarray}
where 
\begin{eqnarray*}
\frac 1n\frac{\partial l_{inc}(\mu,\eta,\theta)\mid \mathbf{y})}{\partial\mu}&=&\frac 1n\sum_{i=1}
\left(z_i+\delta\sqrt{2/{\pi}} \right)\frac 1{\sigma}
-\frac 1n\sum_{i=1}^n\frac{\phi(A_i)}{\Phi(A_i)}\frac{\alpha}{\sigma},\\
\frac 1n\frac{\partial l_{inc}(\mu,\eta,\theta)\mid \mathbf{y})}{\partial\eta}
&=&-1+\frac 1n\sum_{i=1}
\left(z_i+\delta\sqrt{2/{\pi}} \right)z_i
-\frac {\alpha}n\sum_{i=1}^n\frac{\phi(A_i)}{\Phi(A_i)}z_i,\\
\frac 1n\frac{\partial l_{inc}(\mu,\eta,\theta)\mid \mathbf{y})}{\partial\theta}
&=&-\frac 1n\sum_{i=1}
\left(z_i+\delta\sqrt{2/{\pi}} \right)(1-\delta^2)\sqrt{2/{\pi}}\\
&&+\frac 1n\sum_{i=1}^n\frac{\phi(A_i)}{\Phi(A_i)}
\left( \frac{(e^{\theta}+e^{-\theta})}2\left( z_i
+\delta\sqrt{2/{\pi}}\right) 
+\alpha(1-\delta^2)\sqrt{ 2/{\pi}}\right)
\end{eqnarray*}
with $A_i=\alpha(z_i+\delta\sqrt{ 2/{\pi}})$.
To develop the thoery, denote by $C=(c_{ij})_{3\times 3}=C_{(\mu,\eta,\theta)}=(c_{ij}(\mu,\eta,\theta))_{3\times 3}$ the second derivative matrix of the log-likelihood with respect to $(\mu,\eta,\theta)$. 
Applying the Taylor expansion to the functions in (\ref{mleq}) at the ground-truth
$(\mu_0,\eta_0,\theta_0)$, we have
\begin{eqnarray}\label{case1eq}
0&=&\frac 1{\sqrt{n}}\frac{\partial l_{inc}(y_i)}{\partial (\mu_0,\eta_0,\theta_0)^T}
-\frac {\lambda}{2\sqrt{n}}(e^{2\theta_0}-e^{-2\theta_0})\e_3
\nonumber\\
&&+(C-D)_{(\mu^*,\eta^*,\theta^*)}
\sqrt{n}(\hat{\mu}-\mu_0,\hat{\eta}-\eta_0,\hat{\theta}-\theta_0)^T,
\end{eqnarray}
where $ \e_3=(0,0,1)^T$ and $(\mu^*,\eta^*,\theta^*)=(\mu_0,\eta_0, \theta_0)+t(\hat{\mu}-\mu_0,\hat{\eta}-\eta_0,\hat{\theta}-\theta_0)$, $0\le t\le 1.$ 
Denote $D_{\theta\lambda/n}=\mbox{diag}\left(0,0,\frac{\lambda}{n}(e^{2\theta}+e^{-2\theta})\right)$ and $(C-D)_{(\mu,\eta,\theta)}=C_{(\mu,\eta,\theta)}-D_{\theta\lambda/n}.$
For the notation simplicity, let $D_{0\lambda/n}$ denote $D_{\theta_0\lambda/n}$, $C_0$ denote $\lim_{n\to\infty} C_{(\mu_0,\eta_0, \theta_0)}$ and $(C-D)_{0\lambda/n}$ denote $C_0-D_{0\lambda/n}$. Then $C_{(\mu_0,\eta_0, \theta_0)}=C_0 +O_p(1/\sqrt{n}),$ where $-C_0$ is the Fisher information matrix at $(\mu_0,\eta_0, \theta_0) $.
 For a pre-specified positive constant $\omega_0$, consider $0\le \lambda/n\le \omega_0$.  
Taking $\{c_{ij}=c_{ij}(\mu,\eta,\theta): (\mu,\eta,\theta)\in \Omega_1\}$ as empirical processes to which we apply the weak large law,  we have, uniformaly for $0\le \lambda/n\le\omega_0$ and bounded $(\mu_0,\eta_0,\theta_0)$,
\begin{eqnarray*}
\mid\mid (C-D)_{(\mu^*,\eta^*,\theta^*)}-(C-D)_{0\lambda/n}\mid\mid=o_p(1).
\end{eqnarray*}
We have 
\begin{proposition} \label{prop1}
Assume that the MPLE of $\theta$ is in $\Omega_1$ and that the penalty $\mbox{pen}(\theta)$ satisfies the condition $({\it C1})$. Then, when
the true value $\theta_0\not=0$, as $\lambda/\sqrt{n}\to 0$, the MPLE $(\hat{\theta}_{\lambda},\hat{\sigma}_{\lambda},\hat{\theta}_{\lambda})$ is asymptotically optimal in the sense that it is asymptotically unbiased and attains the Cramer-Rao low bound to estimation error.
\end{proposition}

As in practice, the true value $\theta_0$ is unknown, we have to use a data-driven cross-validation to tun the penalty. In the following theorem, we show that $\lambda_{cv}/\sqrt{n}\to 0$ and the cross-validated MPLE is asymptotically optimal in terms of mean square error when the underlying $\theta_0\not=0$. 

\begin{theorem}\label{th1}
 Assume that the MPLE of $\theta$ is in $\Omega_1$ and that the penalty $\mbox{pen}(\theta)$ satisfies the condition $({\it C1})$. Then, when the true value $\theta_0\not=0$, we have
$\lambda_{cv}/\sqrt{n}\to 0$ in probability and the MPLE $(\hat{\theta}_{\lambda_{cv}},\hat{\sigma}_{\lambda_{cv}},\hat{\theta}_{\lambda_{cv}})$ is asymptotically unbiased and attains the Cramer-Rao low bound to estimation error.
\end{theorem}

Let $z_{i0}=(y_i-\mu_0)/\sigma_0.$ In the next proposition, we show that  a fixed $\lambda$ in the Q penalty may give rise to a biased estimate of $\alpha$.
\begin{proposition}\label{prop2}
Assume that the MPLE of $\theta$ is in $\Omega_1$ and that the penalty $\mbox{pen}(\theta)$ satisfies the condition $({\it C1})$. Then, when the true value $\theta_0=0$, we have 
$\hat{\theta}_{\lambda}=0$ for
$
\lambda\ge\max\left\{\sum_{i=1}^n\left(1-z_{i0}^2\right)/{\pi},0\right\}.
$
And on $\sum_{i=1}^n\left(1-z_{i0}^2\right)>0$, $\hat{\theta}_{\lambda_{cv}}$ is non-zero for $0\le\lambda<\sum_{i=1}^n\left(1-z_{i0}^2\right)/{\pi}$.
\end{proposition}
The above proposition implies that when the true value $\theta_0=0$, we have
\begin{enumerate}
 \item [(i)] As $\lambda/\sqrt{n}\to\infty$, we have $\hat{\theta}_{\lambda}=0.$
\item [(ii)] On $\frac 1{\sqrt{n}}\sum_{i=1}^n
\left(1-z_{i0}^2 \right)\le 0 $, for any $\lambda\ge 0,$ we have
$\hat{\theta}_{\lambda}=0.$
\item [(iii)] On $\frac 1{\sqrt{n}}\sum_{i=1}^n
\left(1-z_{i0}^2 \right)>0 $, there is $\lambda$ such that $\frac 1n l_{incp}(\hat{\mu}_{\lambda},\hat{\eta}_{\lambda},\hat{\theta}_{\lambda}|\by)$ attains the maximum at non-zero $\hat{\theta}_{\lambda}.$ 
\end{enumerate}

Let $\hat{\mu}_{j}, \hat{\eta}_j$ be the MLEs of $\mu$ and $\eta$ based on the subsample $\bf{y}_j$ when $\theta$ is known to be zero. Let $\hat{\mu}_{[-j]\lambda},$ 
$\hat{\eta}_{[-j]\lambda}$ and $\hat{\theta}_{[-j]\lambda}$ be the panalised MLEs based on the remaining observations
$(\mathbf{y}_i)_{i\not=j}$ after removing $\mathbf{y}_j$ from $\mathbf{y}.$ For notational simplicity, let $I^{*[-j]}_{11\lambda}$, $I^{*[-j]}_{12\lambda}$, $I^{*[-j]}_{21\lambda}$ and 
 $I^{*[-j]}_{22\lambda}$ denote $I_{11}\mid_{(\mu^*_{[-j]\lambda},\eta^*_{[-j]\lambda},\theta^*_{[-j]\lambda})}$, $I_{12}\mid_{(\mu^*_{[-j]\lambda},\eta^*_{[-j]\lambda},\theta^*_{[-j]\lambda})}$, $I_{21}\mid_{(\mu^*_{[-j]\lambda},\eta^*_{[-j]\lambda},\theta^*_{[-j]\lambda})} $ and 
 \newline $I_{22}\mid_{(\mu^*_{[-j]\lambda},\eta^*_{[-j]\lambda},\theta^*_{[-j]\lambda})}$ respectively.
 Then $I_{11}\mid_{(\mu^*_{[-j]\lambda},\eta^*_{[-j]\lambda},\theta^*_{[-j]\lambda})}=I_{110}(1+o_p(1))$. 
Let $$\lambda_r=\max_{1\le j\le K}\max\{\sum_{i\in [-j]}(1-z_{i0}^2)/\pi,0\}$$ The next theorem shows that if the underlying value of $\theta$ is zero, then the $\cv_a(\lambda)$ attains a local minimum when $\lambda\ge\lambda_r$. 

\begin{theorem}\label{th2}
Assume that the true value $\theta_0=0$, the MPLE of $\theta$ is in $\Omega_1$ and that the penalty $\mbox{pen}(\theta)$ satisfies the condition $({\it C1})$. Then, for 
$\lambda\ge \lambda_r$, the function $\cv_a(\lambda)$ is asymptotically flat, that is
\begin{eqnarray*}
\cv_a(\lambda)&=&-\frac 1K\sum_{j=1}^K
l_{inc}(\hat{\mu}_j,\hat{\eta}_j,0\mid {\mathbf y}_j)
+\frac 1{2K}\sum_{j=1}^Kn_j\left(\frac 1{n_{[-j]}}\sum_{i\in [-j]} \bu_i
-\frac 1{n_{j}}\sum_{i\in [j]}\bu_i\right)^T\\
&&\times (-I_{110})^{-1}
\left(\frac 1{n_{[-j]}}\sum_{i\in [-j]} \bu_i
-\frac 1{n_{j}}\sum_{i\in [j]}\bu_i\right) (1+o_p(1)).
\end{eqnarray*}
\end{theorem}

\section{Penalised EM Algorithm}

We first introduce a positive latent random variable $W$ such that $(Y,W)$ has the following easily calculated \textbf{joint} density
\[
g(y,w)=\frac{2}{\sigma}\phi\left(z+\delta\sqrt{{2}/{\pi}}\right)\phi\left(w-\alpha\left(z+\delta\sqrt{{2}/{\pi}}\right)\right),\qquad y\in\mathbb{R},w\in(0,\infty),
\]
 with marginal density $f(y;\mu,\sigma,\alpha)$ for $Y$, where $z$ denotes $(y-\mu)/\sigma$.
Letting $z_i$ denote $(y_i-\mu)/\sigma(\eta)$ as before and augmenting $\mathbf{y}$ by $\mathbf{w}=(w_{i})_{1\le i\le n}$, we
form the complete data $\left(\mathbf{y},\mathbf{w}\right)=(y_{i},w_{i})_{1\le i\le n}$ with the penalised complete-data log-likelihood 
\begin{eqnarray*}
l_{comp}(\mu,\eta,\theta\mid\mathbf{y},\mathbf{w}) 
 & = &- \frac{n}{2}\log\left({\sigma(\eta)^{2}\pi^{2}}\right)-\frac{1}{2}\sum_{i=1}^{n}\left(z_i+\delta(\theta)\cdot\sqrt{{2}/{\pi}}\right)^{2}\\
 &  & -\frac{1}{2}\sum_{i=1}^{n}\left(w_{i}-\alpha(\theta)\cdot\left(z_i+\delta(\theta)\cdot\sqrt{{2}/{\pi}}\right)\right)^{2}-\lambda(e^{\theta}-e^{-\theta})^2/4.
\end{eqnarray*}
For the $v$-th iteration, let $z^{(v)}_i=({y_{i}-\hat{\mu}^{(v)}})/{\sigma(\hat{\eta}^{(v)})}$ and 
$
b_{i}=z_i+\delta(\theta)\sqrt{{2}/{\pi}}$ and $b_{i}^{(v)}=z^{(v)}_i+\delta(\hat{\theta}^{(v)})\sqrt{{2}/{\pi}}.$ The penalised EM algorithm contains the {\bf E-step} and {\bf M-step} as follows.

{\bf E-Step:} Given the estimates $\hat{\mu}^{(v)}$, $\hat{\eta}^{(v)}$ and $\hat{\theta}^{(v)}$
obtained in the $v$-th iteration, compute the conditional
expectation of the complete log-likelihood:
\begin{eqnarray*}
&&  \Psi(\mu,\eta,\theta\mid\hat{\mu}^{(v)},\hat{\eta}^{(v)},\hat{\theta}^{(v)})
  =  E_{\mathbf{w}|\mathbf{y},\hat{\mu}^{(v)},\hat{\eta}^{(v)},\hat{\theta}^{(v)}}\left[l_{com}(\mu,\eta,\theta\mid\mathbf{y},\mathbf{w})\right]-\lambda(e^{\theta}-e^{-\theta})^2/4\\
& &\qquad=- n\log(\pi)-n\eta-\frac{1}{2}\sum_{i=1}^{n}b_{i}^{2}\\
&&\qquad\quad-\frac{1}{2}\sum_{i=1}^{n}\left\{ 1+\left(\alpha(\hat{\theta}^{(v)})\cdot b_{i}^{(v)}-2\alpha(\theta)\cdot b_{i}\right)\frac{\phi\left(\alpha(\hat{\theta}^{(v)})\cdot b_{i}^{(v)}\right)}{\Phi\left(\alpha(\hat{\theta}^{(v)})\cdot b_{i}^{(v)}\right)}\right.\\
 &  &\qquad\quad \left.+\left(\alpha(\hat{\theta}^{(v)})\cdot b_{i}^{(v)}-\alpha(\theta)\cdot b_{i}\right)^{2}\right\}-\lambda(e^{\theta}-e^{-\theta})^2/4.
\end{eqnarray*}

{\bf M-Step:}
For the simplicity of notation, we denote $\Psi(\mu,\eta,\theta\mid\hat{\mu}^{(v)},\hat{\eta}^{(v)},\hat{\theta}^{(v)})$ by $\Psi$.
To maximise the conditional expectation $\Psi$, compute partial derivatives
of $\Psi$ w.r.t. $(\mu,\eta,\theta)$ and set these derivatives
equal to $0$. We solve the above partial derivative equations by alternative iterations as follows. 

Firstly, fixing $(\eta,\theta)=(\hat{\eta}^{(v)},\hat{\theta}^{(v)})$,
we update $\mu.$ It follows from $\frac{\partial\Psi}{\partial\mu}=0$
that given $\theta=\hat{\theta}^{(v)}$ and $\eta=\hat{\eta}^{(v)}$, the $(v+1)$-th update
of $\mu$,
\begin{eqnarray*}
\hat{\mu}^{(v+1)} & = & \bar{y}+\sigma(\hat{\eta}^{(v)})\cdot\delta(\hat{\theta}^{(v)})\cdot\sqrt{\frac{2}{\pi}}-\frac{\sigma(\hat{\eta}^{(v)})}{n}\cdot\frac{\alpha(\hat{\theta}^{(v)})}{1+\alpha(\hat{\theta}^{(v)})^{2}}\\
 &  & \times\sum_{i=1}^{n}\left\{ \frac{\phi\left(\alpha(\hat{\theta}^{(v)})b_{i}^{(v)}\right)}{\Phi\left(\alpha(\hat{\theta}^{(v)})b_{i}^{(v)}\right)}+\alpha(\hat{\theta}^{(v)})b_{i}^{(v)}\right\} .
\end{eqnarray*}
Secondly, fixing $(\mu,\theta)=(\hat{\mu}^{(v+1)},\hat{\theta}^{(v)})$,
we update $\eta$. It follows from $\frac{\partial\Psi}{\partial\eta}=0$
that
\begin{eqnarray*}
e^{2\eta}-T(\mu,\theta,\boldsymbol{b}^{(v)})\cdot e^{\eta}-\frac{(1+\alpha(\theta)^{2})e^{2\eta}}{n}\sum_{i=1}^{n}z_i^{2} & = & 0
\end{eqnarray*}
with
\begin{eqnarray*}
T(\mu,\theta,\boldsymbol{b}^{(v)}) & = & \left(1+\alpha(\theta)^{2}\right)\cdot\delta(\theta)\cdot\sqrt{\frac{2}{\pi}}e^{\eta}\bar{z}-\frac{\alpha(\theta)}{n}\sum_{i=1}^{n}e^{\eta}z_i\\
 &  & \times\left[\frac{\phi\left(\alpha(\hat{\theta}^{(v)})b_{i}^{(v)}\right)}{\Phi\left(\alpha(\hat{\theta}^{(v)})b_{i}^{(v)}\right)}+\alpha(\hat{\theta}^{(v)})b_{i}^{(v)}\right].
\end{eqnarray*}
Solving the above quadratic equation, 
 we update $\eta$ (and $\sigma=e^{\eta}$) via
\begin{eqnarray*}
\hat{\sigma}^{(v+1)} & = & \sigma(\hat{\eta}^{(v+1)})=e^{\hat{\eta}^{(v+1)}}
  =  \frac{1}{2}T(\hat{\mu}^{(v+1)},\hat{\theta}^{(v)},\boldsymbol{b}^{(v)})\\
&&+\sqrt{\frac{1}{4}T(\hat{\mu}^{(v+1)},\hat{\theta}^{(v)},\boldsymbol{b}^{(v)})^{2}
+\frac{1+\alpha(\hat{\theta}^{(v)})^{2}}{n}\sum_{i=1}^{n}(y_{i}-\hat{\mu}^{(v+1)})^{2}}.
\end{eqnarray*}
Finally, fixing $(\mu,\eta)=(\hat{\mu}^{(v+1)},\hat{\eta}^{(v+1)})$ and letting 
$
f\left(\theta\right)=\frac{\partial\Psi}{\partial\theta},$ $ f^{\prime}\left(\theta\right)=\frac{\partial^{2}\Psi}{\partial\theta^{2}},
$
we obtain $\theta^{(v+1)}$ by using the Newton-Raphson
iteration to solve the equation $f\left(\theta\right)=0$ .

In each update, we need to verify whether the incomplete data likelihood
is increasing. The PEM algorithm iteration alternates between\textbf{ E-step}
and \textbf{M-step} until
\begin{eqnarray*}
\left|\frac{l_{incp}(\hat{\mu}^{(v+1)},\hat{\eta}^{(v+1)},\hat{\theta}^{(v+1)}\mid\mathbf{y})-l_{incp}(\hat{\mu}^{(v)},\hat{\eta}^{(v)},\hat{\theta}^{(v)}\mid\mathbf{y})}{l_{incp}(\hat{\mu}^{(v)},\hat{\eta}^{(v)},\hat{\theta}^{(v)}\mid\mathbf{y})}\right| & < & \varepsilon
\end{eqnarray*}
where $\varepsilon$ is the tolerance with default value of $10^{-8}$. It is easy to prove that the PEM has a non-decreasing property similar to that of the standard EM. 

\section {Numerical results}
In this section, we report the results of simulation studies designed to assess the performance of our cross-validated MPLE and to
compare it to some existing methods (MLE and Q-based MPLE in the R-package SN,  https://CRAN.R-project.org/package=sn) in terms of median bias and standard error of differences between $(\hat{\mu},\hat{\sigma},\hat{\alpha})$ and the ground truth $(\mu_0,\sigma_0,\alpha_0)$. 

\subsection{Behavior of $\lambda_{cv}$}

We first examine the asymptotic behavior of $\lambda_{cv}$ by conducting the following simulation study.

{\bf Setting 1}: Assume that $Y$ follows a skew normal with unknown parameters $(\mu,\sigma,\alpha)$, where the underlying values $(\mu_0,\sigma_0)=(0,1)$  and $\alpha_0\in \{0,2,3,4\}$. We draw a sample of size
$n$ for $Y$ for each combination of $(\alpha_0,n)$, $\alpha_0\in \{0,2,3,4\}$ and $n\in\{50,100,200,300,$ $400,500,600,1000\}$. We repeat this sampling process $20$ times, obtaining $20$ replicates. 
 
We applied the proposed cross-validation procedure to each sample, obtaining the value of $(\hat{\mu},\hat{\sigma},\hat{\alpha})$ and the value of $\lambda_{cv}$. The results are displayed in Figures \ref{alpha0} and \ref{alphanot0}. The results show that
 when $\alpha_0=0$ (that is, the underlying model is a normal), the sample means and variances of these 20 simulated $\lambda_{cv}/n$ tend to a constant ($\approx 0.0035$)  and zero respectively; when $\alpha_0\not=0$ (that is, the underlying model is a skew-normal), both the sample means and variances of these 20 simulated $\lambda_{cv}/\sqrt{n}$ tend to zero. It follows from Chebyshev's inequality that $\lambda_{cv}/n$ tends to a positive constant in probability when the underlying $\alpha_0=0,$ while
  $\lambda_{cv}/\sqrt{n}$ tends to zero in probability when
$\alpha_0\not=0$. Therefore, the numerical results support the theory we develop in the previous section.

{\centerline {[Put Figure \ref{alpha0} here.]}}
\begin{center}
\captionsetup[subfloat]{position=top}
\begin{figure}[hptb]
\subfloat[\scriptsize{\text{$\alpha_0=0$}}]{
\begin{minipage}[t]{0.48\linewidth}
\centering
\includegraphics[width=2in,height=2in]{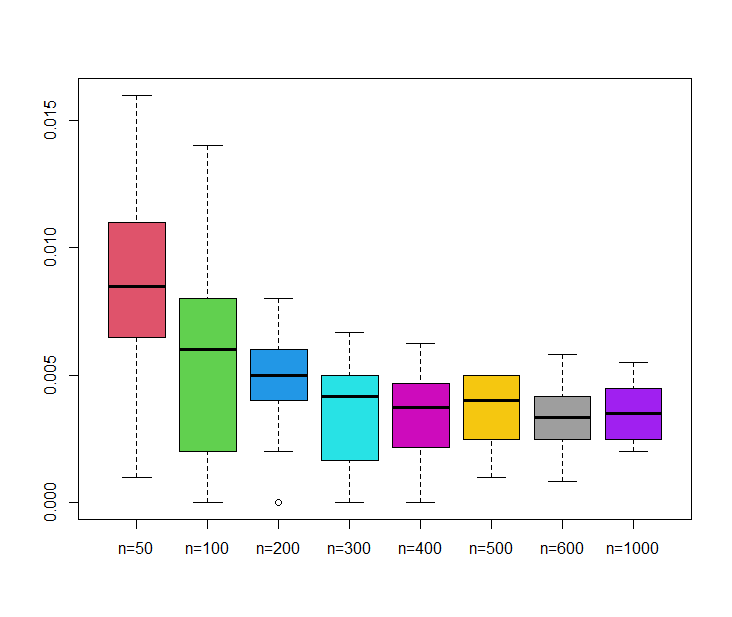}
\end{minipage}
}
\subfloat[\tiny{\text{$\alpha_0=0$}}]{
\begin{minipage}[t]{0.48\linewidth}
\centering
\includegraphics[width=2in,height=2in]{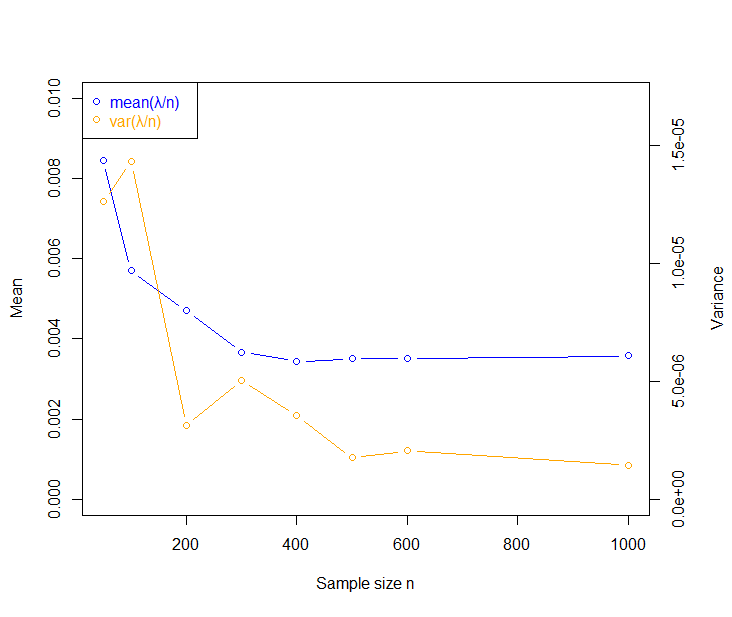}
\end{minipage}
}\\
\caption{Trend plots of $\lambda_{cv}$ when the underlying $(\mu_0,\sigma^2_0,\alpha_0)$ $=(0,1,0)$, $n=50,100,200,$$300,400,500,600$ and  $1000$.
  The box plots on the left, mean-variance chart on the right for $20$ simulated $\lambda_{cv}/n$.
 }
\label{alpha0}
\end{figure}
\end{center}

{\centerline {[Put Figure \ref{alphanot0} here.]}}
\begin{center}
\captionsetup[subfloat]{position=top}
\begin{figure}[hptb]
\subfloat[\scriptsize{\text{$\alpha_0=2$}}]{
\begin{minipage}[t]{0.48\linewidth}
\centering
\includegraphics[width=2in,height=2in]{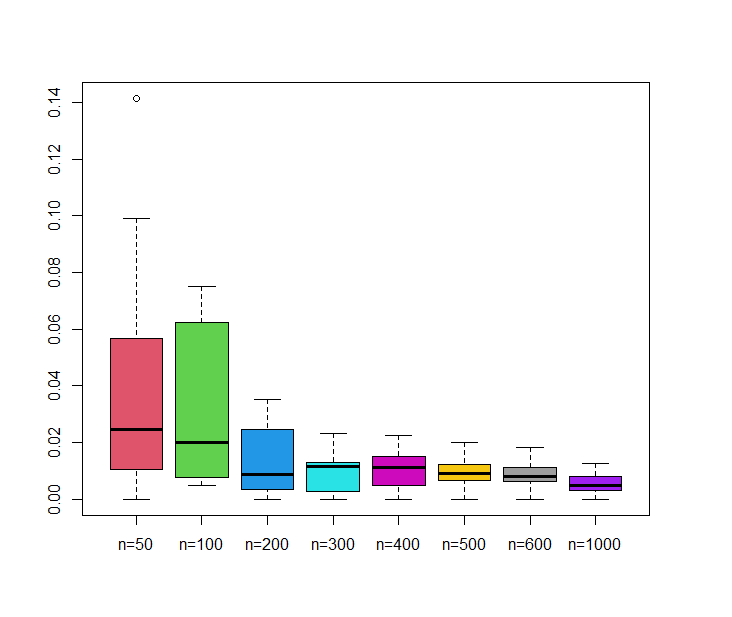}
\end{minipage}
}
\subfloat[\tiny{\text{$\alpha_0=2$}}]{
\begin{minipage}[t]{0.48\linewidth}
\centering
\includegraphics[width=2in,height=2in]{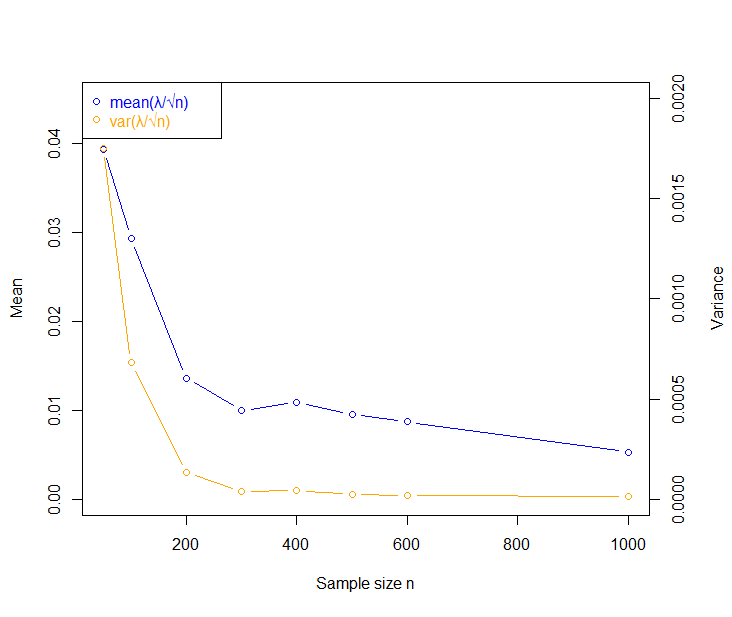}
\end{minipage}
}\\
\subfloat[\scriptsize{\text{$\alpha_0=3$}}]{
\begin{minipage}[t]{0.48\linewidth}
\centering
\includegraphics[width=2in,height=2in]{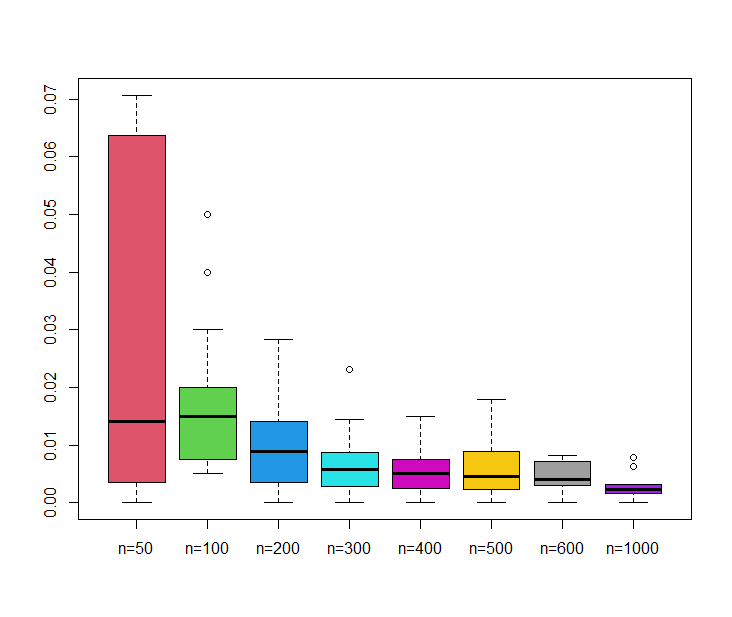}
\end{minipage}
}
\subfloat[\tiny{\text{$\alpha_0=3$}}]{
\begin{minipage}[t]{0.48\linewidth}
\centering
\includegraphics[width=2in,height=2in]{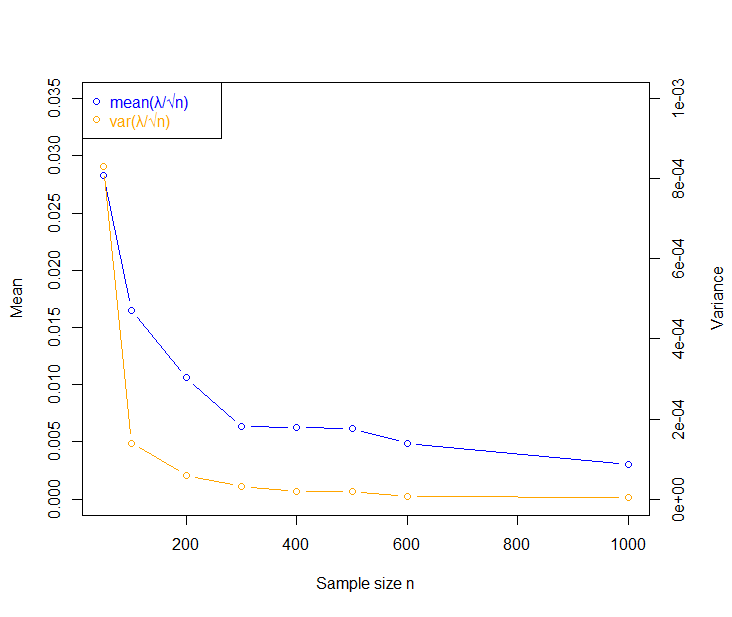}
\end{minipage}
}\\
\subfloat[\scriptsize{\text{$\alpha_0=4$}}]{
\begin{minipage}[t]{0.48\linewidth}
\centering
\includegraphics[width=2in,height=2in]{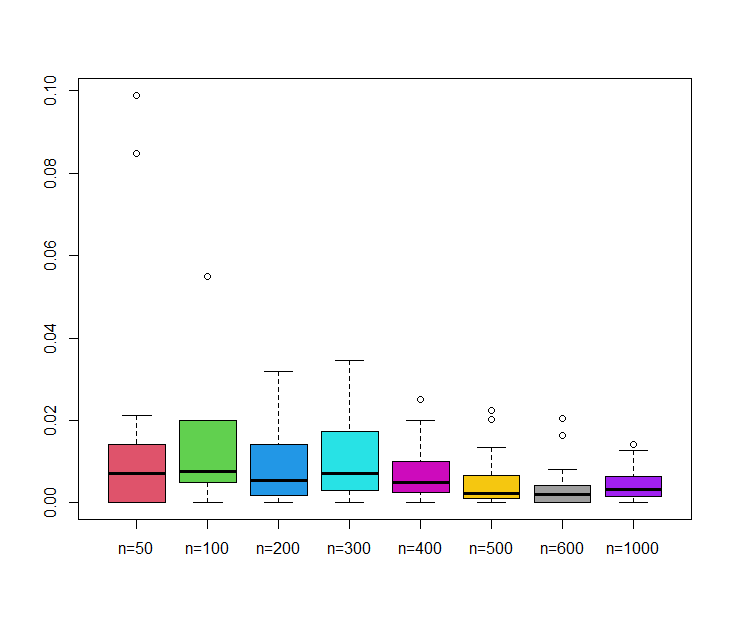}
\end{minipage}
}
\subfloat[\tiny{\text{$\alpha_0=4$}}]{
\begin{minipage}[t]{0.48\linewidth}
\centering
\includegraphics[width=2in,height=2in]{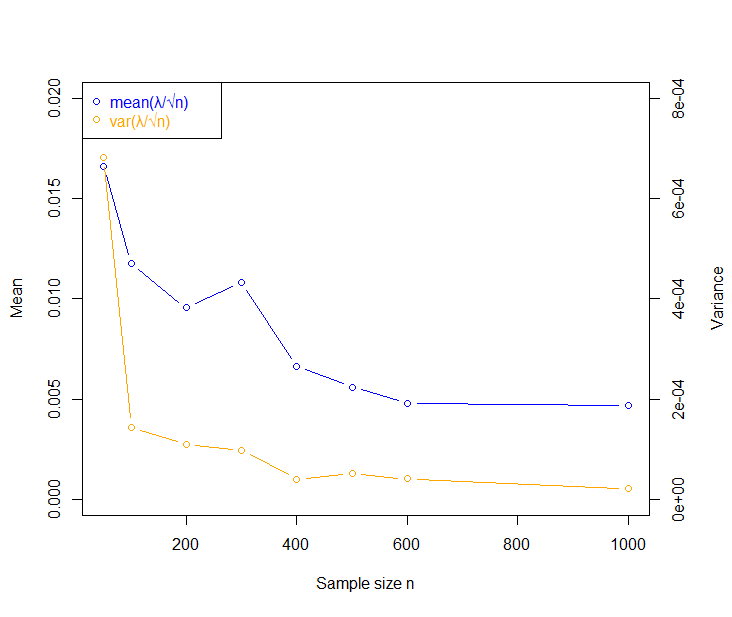}
\end{minipage}
}\\
\caption{Trend plots of $\lambda_{cv}$ when the underlying $(\mu_0,\sigma^2_0)$ $=(0,1)$, $n=50,100,200,$$300,400,500,600$ and $1000$. In each row, the box plots on the left, mean-variance chart on the right for $20$ simulated $\lambda_{cv}/\sqrt{n}$.
 Row $1$ for $\alpha_0=2$ .
Row $2$ for $\alpha_0=3$.
Row $3$ for $\alpha_0=4$.
 }
\label{alphanot0}
\end{figure}
\end{center}

\subsection{Estimation error}
The log-likelihood function is non-quadratic at the
stationary point $\alpha=0$, which makes it non-trivial to estimate. Moreover, the MLE of $\alpha$ is diverging with a positive probability. \cite{r5} proposed Q-based MPLE by maximising the penalised likelihood $l_{p}\left(\theta\right)=l\left(\theta\right)-\lambda\log\left(1+c_{2}\alpha^{2}\right)$
in order to tackle the divergent behavior of estimate $\hat{\alpha}$. Using Firth's bias correction technique, they fixed $\lambda$ and $c_{2}$ as constants with $\lambda\approx0.875913$ and $c_{2}\approx0.856250$. In contrast,
we determine the penalty coefficient $\lambda$ by $10-$fold cross-validation. 
 Intuitively, when the true value $\alpha_0$
is approaching to $0$, the penalty coefficient $\lambda$ should be relatively larger
compared to the case where $\alpha_0$ is away from $0$. Cross-validation chooses the penalty coefficient by letting dataset speak for itself.
In the next simulation study,  we demonstrate that our cross-validated MPLE can outperform the MLE and Q-based MPLE procedures when the underlying value $\alpha_0=0$. 
The result confirms the theory developed in the previous section.

{\bf Setting 2}:  We first generate $\mu_0\sim U(-2,2)$, $\sigma_0\sim U(0.5,1.5)$ and choose $\alpha_0\in \{0,1,2,3,5\}$. Then for each combination of $(\alpha_0,n)$, $\alpha\in \{0,1,2,3,5\}$ and $n\in\{50, 100,$ $200, 400\}$ and given the value of $(\mu_0,\sigma_0,\alpha_0)$, we draw samples of size
$n$ from a skew normal with parameters $(\mu_0,\sigma_0,\alpha_0)$. We repeat this sampling process $m=20$ times, obtaining $m$ replicates. 

{\centerline {[Put Figure \ref{accuracyalp0} here.]}}
\begin{center}
\captionsetup[subfloat]{position=top}
\begin{figure}[hptb]
\subfloat[\scriptsize{\text{$\alpha_0=0$}}]{
\begin{minipage}[t]{0.48\linewidth}
\centering
\includegraphics[width=2in,height=2.in]{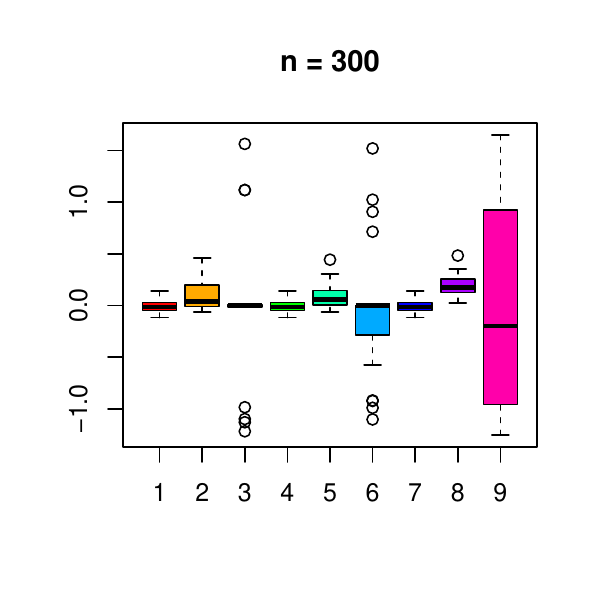}
\end{minipage}
}
\subfloat[\scriptsize{\text{$\alpha_0=0$}}]{
\begin{minipage}[t]{0.48\linewidth}
\centering
\includegraphics[width=2in,height=2.in]{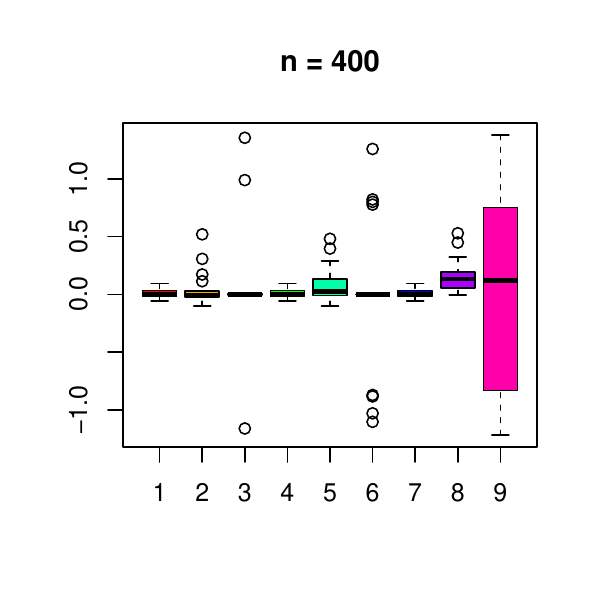}
\end{minipage}
}\\
\subfloat[\scriptsize{\text{$\alpha_0=0$}}]{
\begin{minipage}[t]{0.48\linewidth}
\centering
\includegraphics[width=2in,height=2.in]{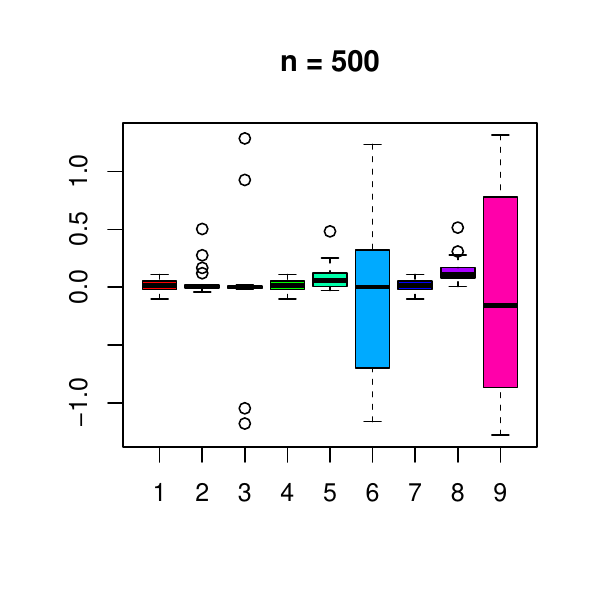}
\end{minipage}
}\hfill
\subfloat[\tiny{\text{$\alpha_0=0$}}]{
\begin{minipage}[t]{0.48\linewidth}
\centering
\includegraphics[width=2in,height=2.in]{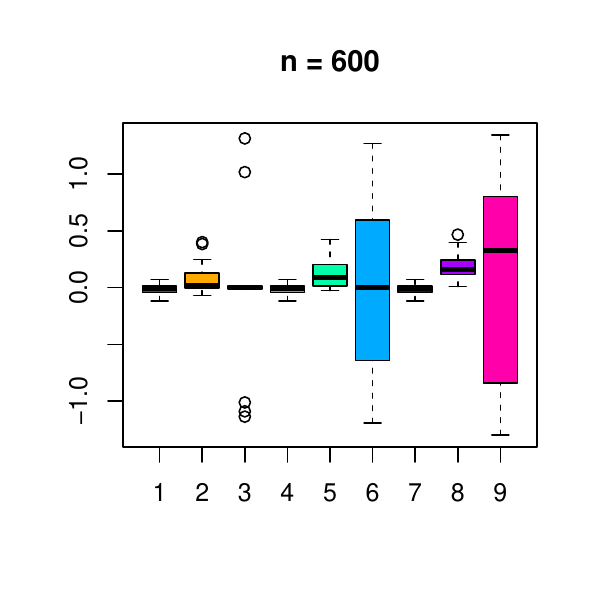}
\end{minipage}
}\\
\subfloat[\scriptsize{\text{$\alpha_0=0$}}]{
\begin{minipage}[t]{0.48\linewidth}
\centering
\includegraphics[width=2in,height=2.in]{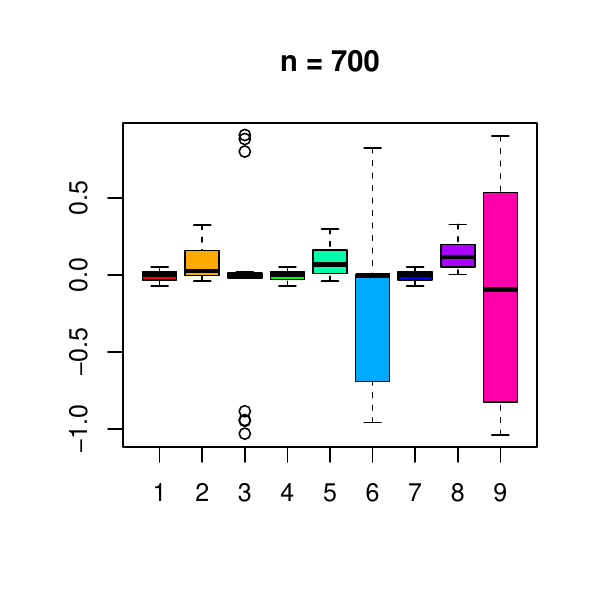}
\end{minipage}
}\hfill
\subfloat[\scriptsize{\text{$\alpha_0=0$}}]{
\begin{minipage}[t]{0.48\linewidth}
\centering
\includegraphics[width=2in,height=2.in]{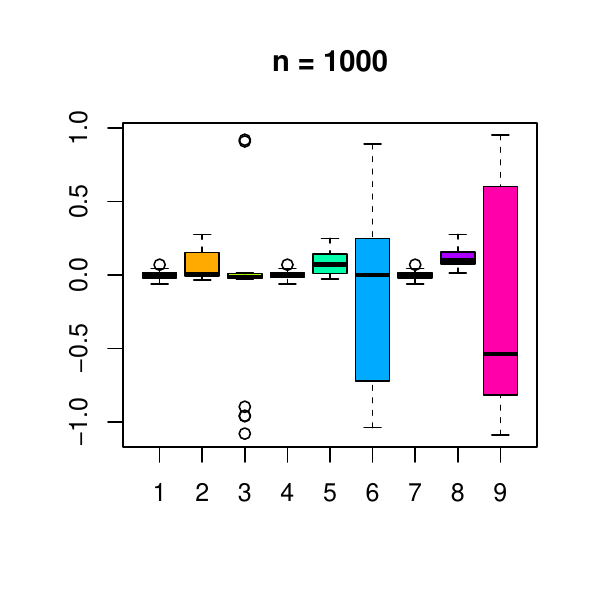}
\end{minipage}
}\\
\caption{Box plots of  $\hat{\mu}-\mu_0$, $\hat{\sigma}-\sigma_0$ and $\hat{\alpha}-\alpha_0$ for $20$ replicates when $\alpha_0=0$, $\mu_0\sim U(-2,2)$, $\sigma_0\sim U(0.5,1.5)$ and $n=300, 400, 500,$ 600,
700 and 1000 respectively. Each panel contains 9 box plots; from the left to right, the first three box plots for location, scale and skewness parameters of the cross-validated MPLE, the second three plots for the Q-based MPLE and the last three for the MLE.}
\label{accuracyalp0}
\end{figure}
\end{center}

In Figure \ref{accuracyalp0} and Figures 2, 3 and 4 in the Online Supplementary Material, the boxplots of estimation errors of $(\mu,\sigma,\alpha)$ suggest that both the proposed cross-validated MPLE and the Q-based MPLE performed substantially better than the MLE in all cases in terms of mean square error. For the sample size $\ge 500$, the cross-validated MPLE does have a strong superiocity over the MPLE, demonstrating that fixing the penalty coefficient to a constant in the Q-based MPLE can compromise the performance of the MPLE. Figures 3 and 4 in the Online Supplementary Material demonstrate that the cross-validated MPLE and MPLE virtually coincide when the underlying value 
$\alpha_0$ is not zero. When $\alpha_0=1$, both the cross-validated MPLE and MPLE tend to shrink to zero, suggesting that weak skewness will be filtered out after the penalisation.

\subsection{IC50 Data}
In this subsection, we considered an IC50 dataset derived from the experiment in
\cite{r14}. We first fit the proposed skew-normal model to log-IC50 measurements of each anti-cancer drugs over 111 cancer cell lines and then characterise these drugs by their estimated location, scale and skewness parameters $(\hat{\mu},\hat{\sigma},\hat{\alpha})$. The results are displayed in Figure \ref{districlusters}.
 There are 170 out
of 227 anti-cancer drugs with evident skewnesses $|\hat{\alpha}|>1$. We applied K-means to these estimates, obtaining 4 clusters with centres
$(-1.47, 2.50	,3.95)$, 
$(2.56,1.50, -0.43),$
$(2.88,1.91,-5.03)$,
$(2.07,2.54,-27.8)$ and of sizes $45,107,65,10$ respectively. 
The distributions in Cluster 4 are very negatively skewed, consisting of cancer growth blockers including inhibitors of LCK, BRAF, C-RAF-1, receptor tyrosine kinase and SRC kinase.  Cluster 3 is positive log-response group, where the distribution of $\hat{\alpha}$ is negatively skewed with $\hat{\mu}$ mainly taking positive values. The distributions in Cluster 2 are close to normal while the distributions in Cluster 1 (resistance group) are positively skewed with log-response mainly taking negative values.  Therefore, classification of drugs to four clusters indicate different patterns of drug response while drugs in the same cluster show a similar mechanism of action.

{\centerline {[Put Figure \ref{districlusters} here.]}}
\begin{center}
\captionsetup[subfloat]{position=top}
\begin{figure}[hptb]
\subfloat[\scriptsize{\text{Scatter plot}}]{
\begin{minipage}[t]{0.8\linewidth}
\centering
\includegraphics[width=2in,height=2.in]{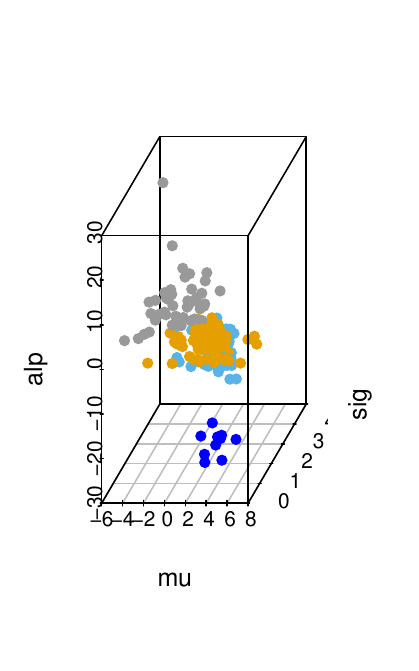}
\end{minipage}
}\\
\subfloat[\scriptsize{\text{Clusters 1 and 3}}]{
\begin{minipage}[t]{0.48\linewidth}
\centering
\includegraphics[width=2in,height=2.in]{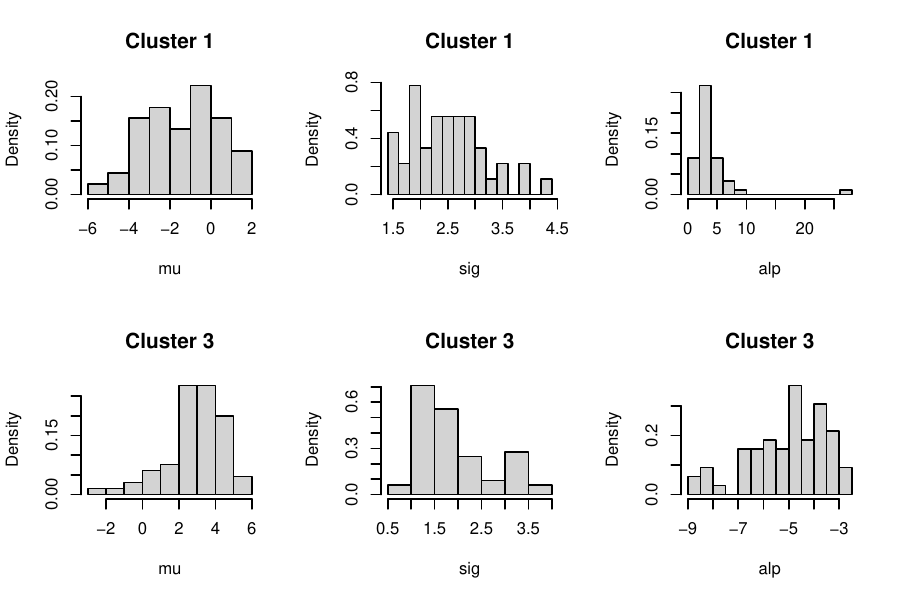}
\end{minipage}
}\hfill
\subfloat[\scriptsize{\text{Clusters 2 and 4}}]{
\begin{minipage}[t]{0.48\linewidth}
\centering
\includegraphics[width=2in,height=2.in]{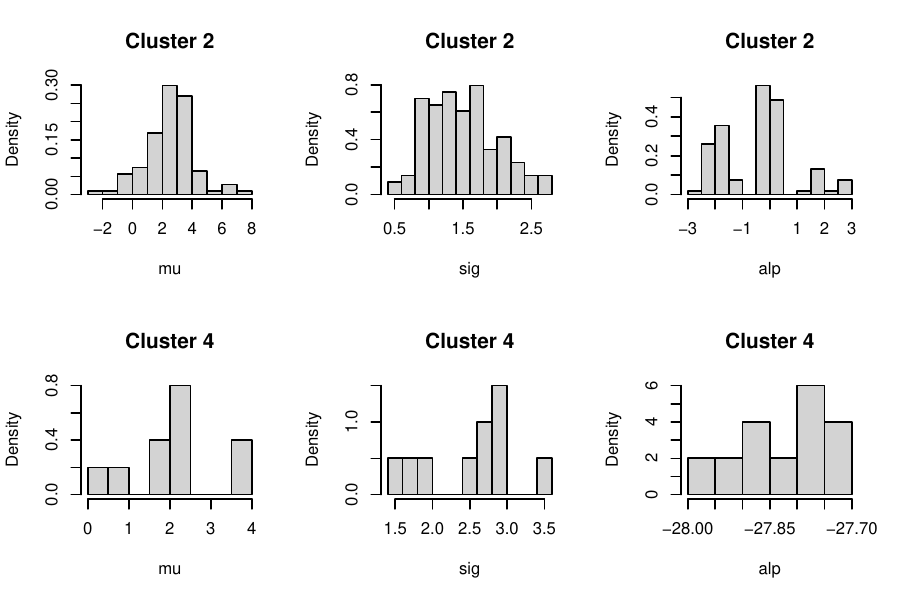}
\end{minipage}
}\\
\caption{\small
{\noindent 
(a) The scatter plot of estimates $(\hat{\mu},\hat{\sigma},\hat{\alpha})$ for the log-IC50 measurements of 227 anti-cancer drugs. Clusters 1, 2, 3, 4 are with gray, orange, green and blue colors respectively.
(b) Distribution patterns of estimated $\hat{\mu}$, $\hat{\sigma}$ and $\hat{\alpha}$ in the clusters 1 and 3. (c) Distribution patterns of estimated $\hat{\mu}$, $\hat{\sigma}$ and $\hat{\alpha}$ in the clusters 2 and 4.
}
}
\label{districlusters}
\end{figure}
\end{center}

Here, the numerical result demonstrated the role of the skewness in discriminating functions of drugs. 
 Given empirical distributions of
the skewness of the log-IC50 in each drug group described as in Figure \ref{districlusters},
the function role of a new drug can be tested again the existing group by
 comparing its estimated
skewness paremeter to the expected value of the
skewness of the existing group.

\section {Discussion and Conclusion}
We have proposed a novel approach for determining penalty coefficients in the maximum penalised likelihood estimation for skew normal distribution families 
 by using the multifold cross-validation. The proposed procedure has addressed the problem of under-regularisation in the Q-based MPLE caused by fixing the penalty coefficient  to a constant.
We have conducted an asymptotic analysis on the behavior of the proposed procedure. In particular, under some regularity condition, we have shown that the cross-validated MPLE can make a sharp improvement over the Q-based approach. This has resulted in the asymptotic efficiency of the proposed estimators. 

We have assessed the performance of the proposed procedure by use of simulated and real data. The simulations have demonstrated that our new procedure can substantially outperform the Q-based MPLE in terms of bias and standard error in a range of scenarios. We have applied the proposed procedure to the analysis of an anti-cancer drug sensitivity dataset, identifying two clusters which have contrasting behavior of drug resistance, one with strong negatively skew drug sensitivity and the other with strong positively skew drug sensitivity. The result is consistent with the existing finding about the role of drug Erlotinib in reducing cancer cell lines resistance to drug
Paclitaxel.

 \begin{acks}[Acknowledgments]
The authors would like to thank Professor Gary Green for his constructive comments on skew distributions that improved the quality of this paper.
\end{acks}

\begin{funding}
The first author was supported by the EPSRC grant EP/X038297/1.
The second author was supported by the GTA PhD scholarship, University of Kent..
\end{funding}

\begin{appendix}
\section{Proofs}\label{appA}

\begin{proof}[Proof of Proposition \ref{prop1}]
 Without loss of generality, assume that $\mbox{pen}(\theta)=\mbox{pen}_1(\theta).$
 Note that when  $\theta_0\not=0$, $C_0$ is invertible which implies
invertibility of $(C-D)^{-1}_{0\lambda/n}$. 
 We have
\begin{eqnarray*}
\sqrt{n}(\hat{\mu}-\mu_0,\hat{\eta}-\eta_0, \hat{\theta}-\theta_0)^T
&=&\frac 1{\sqrt{n}}
(-(C-D)^{-1}_{0\lambda/n})\frac{\partial l_{inc}(y_i)}{\partial (\mu_0,\eta_0,\theta_0)^T}(1+o_p(1))\\
&&+\frac{\lambda}{2\sqrt{n}}(e^{2\theta_0}-e^{-2\theta_0})(C-D)^{-1}_{0\lambda/n}
\e_3(1+o_p(1))
\end{eqnarray*}
which is asymptotically normal with asymptotic mean
\[
M_{0\lambda/n}=(C-D)^{-1}_{0\lambda/n}\frac{\lambda}{2{\sqrt{n}}}(e^{2\theta_0}-e^{-2\theta_0})\e_3
\]
and asymptotic covariance matrix
\begin{eqnarray*}
V_{0\lambda/n}=\left(C-D\right)^{-1}_{0\lambda/n}(-C_0)\left(C-D\right)^{-1}_{0\lambda/n}.
\end{eqnarray*}
Let $d_{0\lambda/n}=\sqrt{\lambda/n}\sqrt{e^{2\theta_0}+e^{-2\theta_0}}\e_3$ and 
\[
C^{-1}_0=\begin{pmatrix}
c^{11}_0&c^{12}_0&c^{13}_0\\
c^{21}_0&c^{22}_0&c^{23}_0\\
c^{31}_0&c^{32}_0&c^{33}_0\\
\end{pmatrix},\quad c^{.3}_0=\begin{pmatrix} c^{13}_0\\ c^{23}_0\\c^{33}_0\end{pmatrix},\quad
c^{3.}_0=(c^{31}_0\quad c^{32}_0 \quad c^{33}_0).
\]
Then $D_{0\lambda/n}=d_{0\lambda/n}d_{0\lambda/n}^T=(\lambda/n)(e^{2\theta_0}+e^{-2\theta_0})\e_3\e_3^T$ and
\begin{eqnarray*}
\left(C-D\right)^{-1}_{0\lambda/n}&=&\left(C_0-d_{0\lambda/n}d_{0\lambda/n}^T\right)^{-1}\\
&=&C^{-1}_0+\frac{C^{-1}_0d_{0\lambda/n}d_{0\lambda/n}^TC^{-1}_0}{1+d_{0\lambda/n}^TC^{-1}_0d_{0\lambda/n}}.
\end{eqnarray*}
\begin{eqnarray*}
V_{0\lambda/n}&=&-\left(C^{-1}_0+\frac{C^{-1}_0d_{0\lambda/n}d_{0\lambda/n}^TC^{-1}_0}{1+d_{0\lambda/n}^TC^{-1}_0d_{0\lambda/n}}\right)C_0
\left(C^{-1}_0+\frac{C^{-1}_0d_{0\lambda/n}d_{0\lambda/n}^TC^{-1}_0}{1+d_{0\lambda/n}^TC^{-1}_0d_{0\lambda/n}}\right)
\\
&=&-\left(C^{-1}_0+\frac{\frac{\lambda}n(e^{2\theta_0}+e^{-2\theta_0})}{1+\frac{\lambda}n(e^{2\theta_0}+e^{-2\theta_0})c^{33}_0}c^{. 3}_0c^{3.}_0\right)C_0\\
&&\times\left(C^{-1}_0+\frac{\frac{\lambda}n(e^{2\theta_0}+e^{-2\theta_0})}{1+\frac{\lambda}n(e^{2\theta_0}+e^{-2\theta_0})c^{33}_0}c^{. 3}_0c^{3.}_0\right)\\
M_{0\lambda/n}&=&\left(C^{-1}_0+\frac{\frac{\lambda}n(e^{2\theta_0}+e^{-2\theta_0})}{1+\frac{\lambda}n(e^{2\theta_0}+e^{-2\theta_0})c^{33}_0}c^{. 3}_0c^{3.}_0\right)\frac{\lambda}{2{\sqrt{n}}}(e^{2\theta_0}-e^{-2\theta_0})\e_3.
\end{eqnarray*}
So, when $\lambda/\sqrt{n}\to 0,$ estimate $(\hat{\mu},\hat{\eta},\hat{\theta})$ is asymptotically unbiased and efficient in the sense that its variance asymptotically achieves the Cramer-Rao lower bound. The proof is completed.
\end{proof}

\begin{proof}[Proof of Theorem \ref{th1}]
Without loss of generality, assume that $n$ is a multiple of $K$ and that $n_1=\cdots=n_K=n/K$. Let $n_{-j}=n-n_j=(K-1)n/K.$ Let $D_{0\lambda/n_{-j}}$ denote $(\lambda/n_{-j})(e^{2\theta_0}+e^{-2\theta_0})\e_3\e_3^T$ and 
$(C-D)_{0\lambda/n_{-j}}$ denote $C_0-D_{0\lambda/n_{-j}}$.
 It follows from the equation (\ref{case1eq}) that
\begin{eqnarray*}
&&(\hat{\mu}_{[-j]\lambda}-\mu_0,
\hat{\eta}_{[-j]\lambda}-\eta_0,
\hat{\theta}_{[-j]\lambda}-\theta_0)^T\\
&&\qquad=\frac 1{n_{-j}}\left(\sum_{i\in [n]}-\sum_{i\in [j]}\right)
(-(C-D)^{-1}_{0\lambda/n_{-j}})\frac{\partial l_{inc}(y_i)}{\partial (\mu_0,\eta_0,\theta_0)^T}
(1+o_p(1))\\
&&\qquad+(C-D)^{-1}_{0\lambda/n_{-j}}\frac{\lambda}{2n_{-j}}(e^{2\theta_0}-e^{-2\theta_0})\e_3
(1+o_p(1))
\end{eqnarray*}
Similarly,
\begin{eqnarray*}
(\hat{\mu}_{[j]}-\mu_0,
\hat{\eta}_{[j]}-\eta_0,
\hat{\theta}_{[j]}-\theta_0)^T
&=&\frac 1{n_j}\sum_{i\in [j]}(-C^{-1}_0)\frac{\partial l_{inc}(y_i)}{\partial (\mu_0,\eta_0,\theta_0)^T}(1+o_p(1)).
\end{eqnarray*}
Consiquently,
\begin{eqnarray*}
&&(\hat{\mu}_{[-j]\lambda}-\hat{\mu}_{[j]},
\hat{\eta}_{[-j]\lambda}-\hat{\eta}_{[j]},
\hat{\theta}_{[-j]\lambda}-\theta_{[j]})^T\\
&&\qquad=\frac{K}{K-1}(-(C-D)^{-1}_{0\lambda/n_{-j}})\frac 1n\sum_{i\in [n]}
\frac{\partial l_{inc}(y_i)}{\partial (\mu_0,\eta_0,\theta_0)^T}
(1+o_p(1))\\
&&\qquad\quad+(C-D)^{-1}_{0\lambda/n_{-j}}\frac {\lambda}{2n_{-j}}(e^{2\theta_0}-e^{-2\theta_0})\e_3(1+o_p(1))\\
&&\qquad\quad-\left(-(C-D)^{-1}_{0\lambda/n_{-j}}\frac 1{n_{-j}}-C^{-1}_0\frac 1{n_j}\right)\sum_{i\in [j]}
\frac{\partial l_{inc}(y_i)}{\partial (\mu_0,\eta_0,\theta_0)^T}
(1+o_p(1))\\
&&\qquad=
\frac{K}{K-1}(-(C-D)^{-1}_{0\lambda/n_{-j}})\frac 1n\sum_{i\in [n]}
\frac{\partial l_{inc}(y_i)}{\partial (\mu_0,\eta_0,\theta_0)^T}
(1+o_p(1))\\
&&\qquad\quad+\frac{\lambda}n\frac K{K-1} (C-D)^{-1}_{0\lambda/n_{-j}}\frac {1}{2}(e^{2\theta_0}-e^{-2\theta_0})(1+o_p(1))\\
&&\qquad\quad-\left(-(C-D)^{-1}_{0\lambda/n_{-j}}\frac 1{K-1}-C^{-1}_0\right)\frac 1{n_j}\sum_{i\in [j]}
\frac{\partial l_{inc}(y_i)}{\partial (\mu_0,\eta_0,\theta_0)^T}
(1+o_p(1))\\
\end{eqnarray*}
Let $U_i^T=\frac{\partial l_{inc}(y_i)}
{\partial (\mu_0,\eta_0,\theta_0)}(-C_0)^{-1/2}$ and 
$W_{0\lambda/n_{-j}}=(-C_0)^{1/2}
(-(C-D)^{-1}_{0\lambda/n_{-j}})(-C_0)^{1/2}.$
Then
\begin{eqnarray*}
&&n_j
(\hat{\mu}_{[-j]\lambda}-\hat{\mu}_{[j]},
\hat{\eta}_{[-j]\lambda}-\hat{\eta}_{[j]},
\hat{\theta}_{[-j]\lambda}-\theta_{[j]})
C_0
(\hat{\mu}_{[-j]\lambda}-\hat{\mu}_{[j]},
\hat{\eta}_{[-j]\lambda}-\hat{\eta}_{[j]},
\hat{\theta}_{[-j]\lambda}-\theta_{[j]})^T\\
&&\quad=-\frac K{(K-1)^2}\left(\frac 1{\sqrt{n}}\sum_{i\in [n]}U_i^T\right)W^2_{0\lambda/n_{-j}}\left(\frac 1{\sqrt{n}}\sum_{i\in [n]}U_i\right)\\
&&\qquad-\frac{2K}{(K-1)^2}\frac{\lambda}{\sqrt{n}}\left(\frac 1{\sqrt{n}}\sum_{i\in [n]}
U_i^T\right)W^2_{0\lambda/n_{-j}}(-C_0)^{-1/2}\frac {1}{2}(e^{2\theta_0}-e^{-2\theta_0})\e_3\\
&&\qquad+\frac{2\sqrt{K}}{K-1}\left(\frac 1{\sqrt{n}}\sum_{i\in [n]}U_i^T\right)
\left(W^2_{0\lambda/n_{-j}}\frac 1{K-1}+W_{0\lambda/n_{-j}}\right)
\left(\frac 1{\sqrt{n_j}}\sum_{i\in [j]} U_i\right)\\
&&\qquad-\frac{\lambda^2}{n}\frac{K}{(K-1)^2}\frac {1}{2}(e^{2\theta_0}-e^{-2\theta_0})\e_3^T(-C_0)^{-1/2}W^2_{0\lambda/n_{-j}}(-C_0)^{-1/2}
\frac {1}{2}(e^{2\theta_0}-e^{-2\theta_0})\e_3\\
&&\qquad+
\frac{2\sqrt{K}}{K-1}\frac{\lambda}{\sqrt{n}}\frac {1}{2}(e^{2\theta_0}-e^{-2\theta_0})\e_3^T(-C_0)^{-1/2}\left(W^2_{0\lambda/n}\frac 1{K-1}+W_{0\lambda/n_{-j}} \right)
\frac 1{\sqrt{n_j}}\sum_{i\in [j]}U_i\\
&&\qquad-\left(\frac 1{\sqrt{n_j}}\sum_{i\in [j]}U_i^T\right)\left(W_{0\lambda/n_{-j}}\frac 1{K-1}+I\right)^2\left(\frac 1{\sqrt{n_j}}\sum_{i\in [j]}U_i\right).
\end{eqnarray*}
Expanding the $j$th validated log-likelihood function 
$l_{inc}(\hat{\mu}_{[-j]\lambda},\hat{\eta}_{[-j]\lambda},\hat{\theta}_{[-j]\lambda}\mid {\bf y}_j) $
at the MLE $(\hat{\mu}_j,\hat{\eta}_j,\hat{\theta}_j)$ of $l_{inc}(\mu, \eta, \theta \mid{\bf y}_j)$, we have
\begin{eqnarray*}
l_{inc}(\hat{\mu}_{[-j]\lambda},\hat{\eta}_{[-j]\lambda},\hat{\theta}_{[-j]\lambda}\mid {\bf y}_j)&=&
l_{inc}(\hat{\mu}_{j},\hat{\eta}_{j},\hat{\theta}_{j}\mid {\bf y}_j)\\
&&+0.5\sqrt{n_j}(\hat{\mu}_{[-j]\lambda}-\hat{\mu}_{j},\hat{\eta}_{[-j]\lambda}-\hat{\eta}_{j},\hat{\theta}_{[-j]\lambda}-\hat{\theta}_{j})\\
&&\times \frac 1{n_j}\frac{\partial^2 l_{inc}(\mu, \eta, \theta\mid {\bf y}_j)}
{\partial (\mu, \eta, \theta) \partial (\mu, \eta, \theta)^T}\mid_{(\hat{\mu}^*_{[-j]\lambda},\hat{\eta}^*_{[-j]\lambda},\hat{\theta}^*_{[-j]\lambda})}\\
&&\times \sqrt{n_j} (\hat{\mu}_{[-j]\lambda}-\hat{\mu}_{j},\hat{\eta}_{[-j]\lambda}-\hat{\eta}_{j},\hat{\theta}_{[-j]\lambda}-\hat{\theta}_{j})^T\\
&=&
l_{inc}(\hat{\mu}_{j},\hat{\eta}_{j},\hat{\theta}_{j}\mid {\bf y}_j)\\
&&+0.5\sqrt{n_j}(\hat{\mu}_{[-j]\lambda}-\hat{\mu}_{j},\hat{\eta}_{[-j]\lambda}-\hat{\eta}_{j},\hat{\theta}_{[-j]\lambda}-\hat{\theta}_{j})(1+o_p(1))\\
&&\times C_0\sqrt{n_j}(\hat{\mu}_{[-j]\lambda}-\hat{\mu}_{j},\hat{\eta}_{[-j]\lambda}-\hat{\eta}_{j},\hat{\theta}_{[-j]\lambda}-\hat{\theta}_{j})^T,
\end{eqnarray*}
where 
\[
(\hat{\mu}^*_{[-j]\lambda},\hat{\eta}^*_{[-j]\lambda},\hat{\theta}^*_{[-j]\lambda})=(\hat{\mu}_{j},\hat{\eta}_{j},\hat{\theta}_{j})+t(\hat{\mu}_{[-j]\lambda}-\hat{\mu}_{j},\hat{\eta}_{[-j]\lambda}-\hat{\eta}_{j},\hat{\theta}_{[-j]\lambda}-\hat{\theta}_{j}), 0\le t\le 1.
\]
Consequently,
\begin{eqnarray*}
\cv_a(\lambda)&+&\frac 1K\sum_{j=1}^Kl_{inc}(\hat{\mu}_{j},\hat{\eta}_{j},\hat{\theta}_{j}\mid {\bf y}_j)\\
&=&-\frac 1K\sum_{j=1}^K0.5\sqrt{n_j}(\hat{\mu}_{[-j]\lambda}-\hat{\mu}_{j},\hat{\eta}_{[-j]\lambda}-\hat{\eta}_{j},\hat{\theta}_{[-j]\lambda}-\hat{\theta}_{j})\\
&&\times C_0(1+o_p(1))\sqrt{n_j}(\hat{\mu}_{[-j]\lambda}-\hat{\mu}_{j},\hat{\eta}_{[-j]\lambda}-\hat{\eta}_{j},\hat{\theta}_{[-j]\lambda}-\hat{\theta}_{j})^T\\
&=&-\frac 1{2K}\sum_{j=1}^K
\sqrt{n_j}(\hat{\mu}_{[-j]\lambda}-\hat{\mu}_{j},\hat{\eta}_{[-j]\lambda}-\hat{\eta}_{j},\hat{\theta}_{[-j]\lambda}-\hat{\theta}_{j})\\
&&\times C_0\sqrt{n_j}(\hat{\mu}_{[-j]\lambda}-\hat{\mu}_{j},\hat{\eta}_{[-j]\lambda}-\hat{\eta}_{j},\hat{\theta}_{[-j]\lambda}-\hat{\theta}_{j})^T(1+o_p(1))\\
&=&\left(\frac{\lambda^2}{n}a_{\lambda/n}+2\frac{\lambda}{\sqrt{n}}b_{\lambda/n}+c_{\lambda/n}\right)(1+o_p(1))\\
&=&\left(a_{\lambda/n}\left(\frac{\lambda}{\sqrt{n}}+b_{\lambda/n}/a_{\lambda/n} \right)^2-b^2_{\lambda/n}/a_{\lambda/n}+c_{\lambda/n}\right)(1+o_p(1)) \\
&\ge&\left(-b^2_{\lambda/n}/a_{\lambda/n}+c_{\lambda/n}\right)(1+o_p(1)),
\end{eqnarray*}
where
\begin{eqnarray*}
a_{\lambda/n}&=&\frac{K}{2(K-1)^2}\frac {1}{4}(e^{2\theta_0}-e^{-2\theta_0})^2
\e_3^T(-C_0)^{-1/2}W^2_{0\lambda/n_{-j}}(-C_0)^{-1/2}\e_3,\\
b_{\lambda/n}&=&\frac 1{2(K-1)}\frac {1}{2}(e^{2\theta_0}-e^{-2\theta_0})\left(\frac 1{\sqrt{n}}\sum_{i\in [n]}
U_i^T\right)W_{0\lambda/n_{-j}}(I-W_{0\lambda/n_{-j}})(-C_0)^{-1/2}\e_3\\
c_{\lambda/n}&=&\left(\frac 1{\sqrt{n}}\sum_{i\in [n]}U_i^T\right)\left(\frac { K-2}{2(K-1)^2}W^2_{0\lambda/n_{-j}}-\frac 1{K-1}W_{0\lambda/n_{-j}} \right)
\left(\frac 1{\sqrt{n}}\sum_{i\in [n]}U_i\right)\\
&&+\frac 1{2K}\sum_{j=1}^K\left(\frac 1{\sqrt{n_j}}\sum_{i\in [j]}U_i^T\right)\left(W_{0\lambda/n_{-j}}\frac 1{K-1}+I\right)^2\left(\frac 1{\sqrt{n_j}}\sum_{i\in [j]}U_i\right).
\end{eqnarray*}
Over $\lambda/\sqrt{n}\in [0,\infty)$, when $\lambda/\sqrt{n}=\max\{0,-b_{\lambda/n}/a_{\lambda/n}\},$  $\cv_a(\lambda)$ asymptotically attains the minimum 
$$-b^2_{\lambda/n}/a_{\lambda/n}+c_{\lambda/n} 
=-\frac 1{2K}\frac{\left(\left(\frac 1{\sqrt{n}}\sum_{i\in [n]}
U_i^T\right)W_{0\lambda/n_{-j}}(I-W_{0\lambda/n_{-j}})(-C_0)^{-1/2}\e_3\right)^2}{\e_3^T(-C_0)^{-1/2}W^2_{0\lambda/n_{-j}}(-C_0)^{-1/2}\e_3}+c_{\lambda/n}
$$
which is independent of $\theta_0$.

 Note that
\begin{eqnarray*}
W_{0\lambda/n_{-j}}&=&I+\frac{\lambda}{n}\frac{K}{K-1}\frac{(e^{2\theta_0}+e^{-2\theta_0})(-C_0)^{-1/2}\e_3\e_3^T(-C_0)^{-1/2}}
{1-\frac{\lambda}{n}\frac{K}{K-1}(e^{2\theta_0}+e^{-2\theta_0})\e_3^T(-C_0)^{-1}\e_3}.\\
W_{0\lambda/n_{-j}}(I-W_{0\lambda/n_{-j}})&=&\frac{\lambda}{n}O(1).
\end{eqnarray*}

 For fixed $\theta_0\not=0,$ $\lambda$ satisfying $\lambda/\sqrt{n}\to\infty$ and $0\le \lambda/n\le \omega_0$, it follows from the above equations that $\cv_a(\lambda)$ tends to infinity. While for bounded $\lambda/\sqrt{n}$, $\lambda/n$ tends to zero,
$W_{0\lambda/n_{-j}}\to I$ and 
\begin{eqnarray*}
a_{\lambda/n}&\to& \frac{K}{2(K-1)^2}\frac {1}{4}(e^{2\theta_0}-e^{-2\theta_0})^2
\e_3^T(-C_0)^{-1}\e_3,\\
b_{\lambda/n_{-j}}/a_{\lambda/n_{-j}}&=&-\frac{\lambda}{\sqrt{n}}O(1)\frac{2(e^{2\theta_0}+e^{-2\theta_0})}{\sqrt{n}(e^{2\theta_0}-e^{-2\theta_0})}\\
&&\times\frac{\frac 1{\sqrt{n}}\sum_{i\in [n]}U_i^T(-C_0)^{-1/2}\e_3\e_3^T(-C_0)^{-1/2}}{1-(\lambda K/(n(K-1)))(e^{2\theta_0}+e^{-2\theta_0})\e_3^T(-C_0)^{-1}\e_3}.\\
\left(\frac{\lambda}{\sqrt{n}}+b_{\lambda/n_{-j}}/a_{\lambda/n_{-j}}\right)^2&=&\frac{\lambda^2}{n}(1-o_p(1))^2.\\
c_{\lambda/n}&\to& \frac K{2(K-1)^2}(-\chi^2_{3}+\chi^2_{3K}),
\end{eqnarray*}
where $\chi^2_3$ and $\chi^2_{3K}$ are two dependented chi-squared random variables.
Therefore, for $\lambda/n\in [0,\omega_0]$ $\cv_a(\lambda)$ asymptotically attains the minimum $ -\frac K{2(K-1)^2}\chi^2_{3}+\frac K{2(K-1)^2}\chi^2_{3K}$ when $\lambda/\sqrt{n}= 0.$ This implies that $\lambda_{cv}/\sqrt{n}$ tends to zero in probability when the true value of $\theta_0\not=0.$ The proof is completed.
\end{proof}

\begin{proof}[Proof of Proposition \ref{prop2}]
 Note that when $\theta_0=0$,  $C_0=\mbox{diag}(\sigma^{-2}_0,-2,0)$
 is degenerate. Let 
\[
\bu_i=\left(z_{i0}, z_{i0}^2-1\right)^T.
\]
 Then
\begin{eqnarray*}
\frac 1n\frac{\partial l_{inc}(\mu_0,\eta_0,0)}{\partial (\mu,\eta,\theta)^T}&=&
\frac 1n\sum_{i=1}^n
(\bu_i^T,0)^T
\end{eqnarray*}
And let $z_i^*=(y_i-\mu^*_{\lambda})/\sigma^*_{\lambda},$ we have
\begin{eqnarray*}
c_{11}\mid_{(\mu^*_{\lambda},\eta^*_{\lambda},\theta^*_{\lambda})}&=&-\frac 1{\sigma^{*2}_{\lambda}}-\frac{\alpha^{*2}_{\lambda}}{\sigma^{*2}_{\lambda}}\frac 2{\pi}(1+o_p(1))\\
c_{12}\mid_ {(\mu^*_{\lambda},\eta^*_{\lambda},\theta^*_{\lambda})}
&=&-\frac 2{\sigma^*_{\lambda}}\frac 1n\sum_{i=1}^nz_i^*
+\frac{\alpha^*_{\lambda}}{\sigma^*_{\lambda}}\sqrt{2/{\pi}}O_p(\theta^*_{\lambda}+1/\sqrt{n}).\\
c_{13}\mid_{(\mu^*_{\lambda},\eta^*_{\lambda},\theta^*_{\lambda})}
&=&-\frac{\delta^{*2}_{\lambda}}{\sigma^*_{\lambda}}
\sqrt{2/{\pi}}-\frac {\theta^{*2}_{\lambda}}{2\sigma^*_{\lambda}}\sqrt{2/{\pi}}\\
&&-\frac{\alpha^*_{\lambda}}{\sigma^*_{\lambda}}\frac 1n\sum_{i=1}^n\frac{\phi(A^*_i)}{\Phi(A^*_i)}\left(\alpha^*_{\lambda}\left(z_i^*+\delta^*_{\lambda}\sqrt{2/{\pi}}\right)+\frac{\phi(A^*_i)}{\Phi(A^*_i)} \right)\\
&&\quad \times 
\left(\alpha^*_{\lambda}(1-\delta^{*2}_{\lambda})\sqrt{2/{\pi}} \right)\\
&=&O_p(\theta^{*2}_{\lambda}).
\end{eqnarray*}
\begin{eqnarray*}
\frac{c_{13}|_{(\mu^*_{\lambda},\eta^*_{\lambda},\theta^*_{\lambda})}}{\theta^{*2}_{\lambda}}&\to& 
-\frac 1{\sigma^*_{\lambda}}\sqrt{ 2/{\pi}}\left(\frac 32+2/{\pi}\right) \mbox{  as $\theta^*_{\lambda} \to 0$.}\\
c_{21}\mid_{(\mu^*_{\lambda},\eta^*_{\lambda},\theta^*_{\lambda})}&=&c_{12}\mid_ {(\mu^*_{\lambda},\eta^*_{\lambda},\theta^*_{\lambda})}\\
c_{22}\mid_{(\mu^*_{\lambda},\eta^*_{\lambda},\theta^*_{\lambda})}
&=&\theta^*_{\lambda}\frac 4{\pi}\sum_{i=1}^nz_i^{*2} +
\theta^*_{\lambda}O_p(\theta^*_{\lambda}+1/\sqrt{n}).\\
\frac{c_{23}|_{(\mu^*_{\lambda},\eta^*_{\lambda},\theta^*_{\lambda})}}{\theta^*_{\lambda}}&\to&\frac 4{\pi} \mbox{  as $n\to \infty.$ }\\
c_{31}\mid_{(\mu^*_{\lambda},\eta^*_{\lambda},\theta^*_{\lambda})}&=&c_{13}\mid_{(\mu^*_{\lambda},\eta^*_{\lambda},\theta^*_{\lambda})},\quad
c_{32}\mid_{(\mu^*_{\lambda},\eta^*_{\lambda},\theta^*_{\lambda})}=c_{23}\mid_{(\mu^*_{\lambda},\eta^*_{\lambda},\theta^*_{\lambda})}\\
c_{33}\mid_{(\mu^*_{\lambda},\eta^*_{\lambda},\theta^*_{\lambda})}
&=&
\frac 2{\pi}\frac 1n\sum_{i=1}^n\left(1- z_i^{*2}\right)
+\theta^*_{\lambda}O_p(\theta^{*}_{\lambda}+1/\sqrt{n}).
\end{eqnarray*}

Let
\begin{eqnarray*}
I_{11}=\begin{pmatrix}
c_{11}&c_{12}\\
c_{21}&c_{22}
\end{pmatrix},
I_{12}=\begin{pmatrix}
c_{13}\\
c_{23}
\end{pmatrix},
I_{21}=I_{12}^T,
I_{22}=c_{33}-\frac{\lambda}{n}(e^{2\theta}+e^{-2\theta}).
\end{eqnarray*}

Let $I_{110}=\mbox{ diag}(-1/\sigma_0^2,-2),$ $I_{220}=\frac 2{\pi}\frac 1n\sum_{i=1}^n(1-z_{i0}^2)-2\lambda/n,$  $I^*_{11\lambda}=$ $I_{11}\mid_{(\mu^*_{\lambda},\eta^*_{\lambda},\theta^*_{\lambda})} $, 
\newline $I^*_{12\lambda}=$ $ I_{12}\mid_{(\mu^*_{\lambda},\eta^*_{\lambda},\theta^*_{\lambda})} $, $I^*_{21\lambda}=I_{21}\mid_{(\mu^*_{\lambda},\eta^*_{\lambda},\theta^*_{\lambda})} $ and $I^*_{22\lambda}=I_{22}\mid_{(\mu^*_{\lambda},\eta^*_{\lambda},\theta^*_{\lambda})} $. Then
\begin{eqnarray*}
I^{-1}_{11}&=&\frac 1{c_{11}c_{22}-c_{21}c_{12}}
\begin{pmatrix}
c_{22}&-c_{12}\\
-c_{21}&c_{11}
\end{pmatrix},\\
I_{21}I^{-1}_{11}I_{12}&=&\frac{c^2_{31}c_{22}-2c_{32}c_{21}c_{13}+c^2_{32}c_{11}}{c_{11}c_{22}-c_{21}c_{12}}.\\
I^*_{21\lambda}I^{-1}_{110}I^*_{12\lambda}&=&\frac{\theta^{*2}_{\lambda}O_p(1)}{\frac 1{\sigma^{*2}_{\lambda}}\frac 2n
\sum_{i=1}^nz_i^{*2}+O_p(1/n)+\theta^*_{\lambda}O_p(\theta^*+1/\sqrt{n})}.\\
I^*_{22\lambda}-I^*_{21\lambda}I^{-1}_{110}I^*_{12\lambda}&=&\frac 2{\pi}\frac 1n\sum_{i=1}^n
\left(1-z_i^{*2}\right)-\frac{2\lambda}{n}
+\theta^*O_p(\theta^*_{\lambda}+1/\sqrt{n}).\\
I^*_{12\lambda}&=&\frac 4{\pi}\theta^*_{\lambda}\left( O_p(\theta^*_{\lambda}),\frac 1n\sum_{i=1}^nz_i^{*2}+O_p(\theta^*_{\lambda}+1/\sqrt{n}) \right)^T.
\end{eqnarray*}
 We have
\begin{eqnarray*}
-\frac 1{\sqrt{n}}\sum_{i=1}^n\bu_i
&=&I^*_{11\lambda}
\begin{pmatrix}
\sqrt{n}(\hat{\mu}_{\lambda}-\mu_0)\\
\sqrt{n}(\hat{\eta}_{\lambda}-\eta_0)
\end{pmatrix}+I^*_{12\lambda}\sqrt{n}\hat{\theta}_{\lambda}\\
0&=&I^*_{21\lambda}\begin{pmatrix}
\sqrt{n}(\hat{\mu}_{\lambda}-\mu_0)\\
\sqrt{n}(\hat{\eta}_{\lambda}-\eta_0)
\end{pmatrix}+I^*_{22\lambda}\hat{\theta}_{\lambda}
\end{eqnarray*}
which implies
\begin{eqnarray*}
-\frac 1{\sqrt{n}}\sum_{i=1}^n\bu_i
-\sqrt{n}\hat{\theta}_{\lambda}I^*_{12\lambda}
&=&I_{110}
\begin{pmatrix}
\sqrt{n}(\hat{\mu}_{\lambda}-\mu_0)\\
\sqrt{n}(\hat{\eta}_{\lambda}-\eta_0)
\end{pmatrix}(1+o_p(1))\\
(-I_{110})^{-1}\frac{1}{\sqrt{n}}\sum_{i=1}^n\bu_i
+\sqrt{n}\hat{\theta}_{\lambda}(-I_{110})^{-1}I^*_{12\lambda}
&=&
\begin{pmatrix}
\sqrt{n}(\hat{\mu}-\mu_0)\\
\sqrt{n}(\hat{\eta}-\eta_0)
\end{pmatrix}(1+o_p(1)),
\end{eqnarray*}
\begin{eqnarray*}
0&=&I^*_{21\lambda}(-I_{110})^{-1}\frac 1n\sum_{i=1}^n\bu_i
+\left(I^*_{21\lambda}(-I_{110})^{-1}I^*_{12\lambda}+I^*_{22\lambda}\right)\hat{\theta}_{\lambda}
\end{eqnarray*}
to which $\hat{\theta}_{\lambda}=0$ is a solution, since $I_{21}\mid_{(\mu^*_{\lambda},\eta^*_{\lambda},0)}=0$ and 
$I_{12}\mid_{(\mu^*_{\lambda},\eta^*_{\lambda},0)}=0$. Note that
\begin{eqnarray*}
\frac 1n l_{incp}(\hat{\mu}_{\lambda},\hat{\eta}_{\lambda},\hat{\theta}_{\lambda}|\by)&=&\frac 1n l_{incp}(\mu_0,\eta_0,0|\by)+\frac 1n\sum_{i=1}^n\bu_i^{\tau}
(\hat{\mu}_{\lambda}-\mu_0,\hat{\eta}_{\lambda}-\eta_0)^{\tau}\\
&&+\frac 12(\hat{\mu}_{\lambda}-\mu_0,\hat{\eta}_{\lambda}-\eta_0,\hat{\theta}_{\lambda})\mbox{diag}(I_{110},I_{220})(\hat{\mu}_{\lambda}-\mu_0,\hat{\eta}_{\lambda}-\eta_0,\hat{\theta}_{\lambda})^{\tau}\\
&&\times (1+o_p(1))\\
&=&\frac 1n l_{incp}(\mu_0,\eta_0,0|\by)+\frac 1n\sum_{i=1}^n\bu_i^{\tau}
(\hat{\mu}_{\lambda}-\mu_0,\hat{\eta}_{\lambda}-\eta_0)^{\tau}\\
&&+\frac 12(\hat{\mu}_{\lambda}-\mu_0,\hat{\eta}_{\lambda}-\eta_0)I_{110}(\hat{\mu}_{\lambda}-\mu_0,\hat{\eta}_{\lambda}-\eta_0)^{\tau}(1+o_p(1))\\
&&+\frac 12 I_{220}\hat{\theta}^2_{\lambda}(1+o_p(1))\\
&=&\frac 1n l_{incp}(\mu_0,\eta_0,0|\by)+\frac 1n\sum_{i=1}^n\bu_i^{\tau}(-I_{110})^{\tau}\frac 1n\sum_{i=1}^n\bu_i(1+o_p(1))\\
&&+\frac 1n\sum_{i=1}^n\bu_i^{\tau}(-I_{110})^{-1}I^*_{12\lambda}\hat{\theta}_{\lambda}(1+o_p(1))\\
&&-\frac 12\left(\frac 1n\sum_{i=1}^n\bu_i^{\tau}(-I_{110})^{-1}+I^*_{21\lambda}(-I_{110})^{-1}\hat{\theta}_{\lambda}\right)(-I_{110})\\
&&\qquad\times\left ((-I_{110})^{-1}\frac 1n\sum_{i=1}^n\bu_i
 +(-I_{110})^{-1}I^*_{12\lambda}\hat{\theta}_{\lambda} \right)\\
&&+\frac 12 I_{220}\hat{\theta}^2_{\lambda}(1+o_p(1))\\
&=&\frac 1n l_{incp}(\mu_0,\eta_0,0|\by)+\frac 1n\sum_{i=1}^n\bu_i^{\tau}(-I_{110})^{\tau}\frac 1n\sum_{i=1}^n\bu_i(1+o_p(1))\\
&&-\frac 12\left(-I_{220}+I_{21*}(-I_{110})^{-1}I^*_{12\lambda}\right)\hat{\theta}^2_{\lambda}(1+o_p(1))
\end{eqnarray*}
which attains the maximum  at $\hat{\theta}_{\lambda}=0$ when $I_{220}\le 0,$ that is, when 
\[
\frac{\lambda}{\sqrt{n}}\ge\max\left\{\frac 1{\pi}\frac 1{\sqrt{n}}\sum_{i=1}^n\left(1-z_{i0}^2\right),0\right\}.
\]
This implies:
\begin{itemize}
 \item When $\lambda/\sqrt{n}\to\infty$, we have $\hat{\theta}_{\lambda}=0.$
\item When $\frac 1{\pi}\frac 1{\sqrt{n}}\sum_{i=1}^n
\left(1-z_{i0}^2 \right)\le 0 $, for any $\lambda\ge 0,$ we have
$\hat{\theta}_{\lambda}=0.$
\end{itemize}
When $\hat{\theta}_{\lambda}=0$, we have
\begin{eqnarray*}
\begin{pmatrix}
\sqrt{n}(\hat{\mu}-\mu_0)/\sigma_0\\
\sqrt{n}(\hat{\eta}-\eta_0)
\end{pmatrix}.
=\begin{pmatrix}
\frac 1{\sqrt{n}}\sum_{i=1}^nz_{i0}\\
\frac 1{2\sqrt{n}}\sum_{i=1}^n\left(z_{i0}^2-1\right)
\end{pmatrix}(1+o_p(1))
\end{eqnarray*}
which is asymptotically normal with mean zero and covariance matrix $\mbox{diag}(1,1/2).$ 
When 
\[
0\le \frac{\lambda}{\sqrt{n}}<\frac 1{\pi}\frac 1{\sqrt{n}}\sum_{i=1}^n
\left(1-z_{i0}^2 \right)-\frac {\sqrt{n}}{2}I^*_{21\lambda}(-I_{110})^{-1}I^*_{12\lambda},
\]
we have
\begin{eqnarray*}
\frac 1n l_{incp}(\hat{\mu}_{\lambda},\hat{\eta}_{\lambda},\hat{\theta}_{\lambda}|\by)>
\frac 1n l_{incp}(\mu_0,\eta_0,0|\by)
+\frac 1n\sum_{i=1}^n\bu_i^{\tau}(-I_{110})^{\tau}\frac 1n\sum_{i=1}^n\bu_i(1+o_p(1)).
\end{eqnarray*}
Therefore, when $\frac 1{\pi}\frac 1{\sqrt{n}}\sum_{i=1}^n
\left(1-z_{i0}^2 \right)>0 $, there is $\lambda$ such that $\frac 1n l_{incp}(\hat{\mu}_{\lambda},\hat{\eta}_{\lambda},\hat{\theta}_{\lambda}|\by)$ attains the maximum at non-zero $\hat{\theta}_{\lambda}$ satisfying
\[
\frac {\sqrt{n}}{2}I^*_{21\lambda}(-I_{110})^{-1}I^*_{12\lambda}<\frac 1{\pi}\frac 1{\sqrt{n}}\sum_{i=1}^n
\left(1-z_{i0}^2 \right).
\]
The proof is completed.
\end{proof}

\begin{proof}[Proof of Theorem \ref{th2}]

It follows from
\[
\frac{\partial l_{inc}(\hat{\mu}_{[-j]\lambda},\hat{\eta}_{[-j]\lambda},\hat{\theta}_{[-j]\lambda}\mid \bf{y}_{[-j]})}{\partial(\mu,\eta,\theta)^T}=0
\]
and Taylor's theorem that
\begin{eqnarray*}
-\frac 1{n_{[-j]}}\sum_{i\in [-j]}
(\bu_i^T,0)^T&=&
\begin{pmatrix}
I_{11}&I_{12}\\
I_{21}&I_{22}
\end{pmatrix}_{(\mu^*_{[-j]\lambda},\eta^*_{[-j]\lambda},\theta^*_{[-j]\lambda})}
\begin{pmatrix}
\hat{\mu}_{[-j]\lambda}-\mu_0\\
\hat{\eta}_{[-j]\lambda}-\eta_0\\
\hat{\theta}_{[-j]\lambda}\\
\end{pmatrix},
\end{eqnarray*}
where $(\mu^*_{[-j]\lambda},\eta^*_{[-j]\lambda},\theta^*_{[-j]\lambda})
=t(\hat{\mu}_{[-j]\lambda},\hat{\eta}_{[-j]\lambda},\hat{\theta}_{[-j]\lambda})$ for some $0\le t\le 1.$ Furthermore, if $\hat{\theta}_{[-j]\lambda}$ $\not=0,$ then 
\[
\frac{I^{*[-j]}_{21\lambda}}{\hat{\theta}_{[-j]\lambda}}(-I_{110})^{-1}\frac 1{n_{[-j]}}\sum_{i\in [-j]}\bu_i
+I^{*[-j]}_{21\lambda}(-I_{110})^{-1}I^{*[-j]}_{12\lambda}+I^{*[-j]}_{22\lambda}=0.
\]

Consider $\lambda$ at which $I^{*[-j]}_{22\lambda}-I^{*[-j]}_{21\lambda}I^{-1}_{110}I_{12\lambda}\not=0.$ 
We have
\begin{eqnarray}\label{cvest}
(\hat{\mu}_{[-j]\lambda}-\mu_0,
\hat{\eta}_{[-j]\lambda}-\eta_0)^T
&=&(-I_{110})^{-1}\left(\frac 1{n_{[-j]}}\sum_{i\in [-j]} \bu_i
+I^{*[-j]}_{12\lambda}\hat{\theta}_{\lambda}\right)(1+o_p(1))\nonumber\\
&=&((-I_{110})^{-1}-(-I_{110})^{-1}I^{*[-j]}_{12\lambda}(I^{*[-j]}_{22\lambda}+I^{*[-j]}_{21\lambda}(-I_{110})^{-1}I_{12*})^{-1}\nonumber\\
&&\times I_{21*}(-I_{110})^{-1})
\frac 1{n_{[-j]}}\sum_{i\in [-j]} \bu_i (1+o_p(1)).
\end{eqnarray}
Similarly, it follows from
\[
\frac{\partial l_{inc}(\hat{\mu}_{j},\hat{\eta}_{j},0\mid \bf{y}_j)}{\partial(\mu,\eta,\theta)^T}=0
\]
and Taylor's theorem that
\[
(\hat{\mu}_{j}-\mu_0,
\hat{\eta}_{j}-\eta_0)^T=(-I_{110})^{-1}\frac 1{n_{j}}\sum_{i\in [j]}u_i(1+o_p(1)).
\]
Consequently, we have
\begin{eqnarray*}
&&(\hat{\mu}_{[-j]\lambda}-\hat{\mu}_j,
\hat{\eta}_{[-j]\lambda}-\hat{\eta}_j)^T\\
&&\qquad=(-I_{110})^{-1}\left(\frac 1{n_{[-j]}}\sum_{i\in [-j]}\bu_i+I^{*[-j]}_{12\lambda}\hat{\theta}_{\lambda}-\frac 1{n_{j}}\sum_{i\in [j]}\bu_i\right)(1+o_p(1))\\
&&\qquad=\frac 1{n_{[-j]}}\sum_{i\in [-j]}
\{(-I_{110})^{-1}-(-I_{110})^{-1}I^{*[-j]}_{12\lambda}\\
&&\qquad\quad\times \left(I^{*[-j]}_{22\lambda}+I^{*[-j]}_{21\lambda}(-I_{110})^{-1}I^{*[-j]}_{12\lambda} \right)^{-1}I^{*[-j]}_{21\lambda}(-I_{110})^{-1}
\} \bu_i (1+o_p(1))\\
&&\qquad\quad-\frac 1{n_j}\sum_{i\in [j]}(-I_{110})^{-1}\bu_i(1+o_p(1)).
\end{eqnarray*}
It follows from Propositin \ref{prop2} that for $
\lambda\ge\max\left\{\frac 1{\pi}\sum_{i\in [-j]}\left(1-z_{i0}^2\right),0\right\},$ we have
 $\hat{\theta}_{[-j]\lambda}=0$, $1\le j\le K$.

On other hand, using the Taylor expansion,  we have
\begin{eqnarray}\label{diff2}
&&l_{inc}(\hat{\mu}_{[-j]\lambda},\hat{\eta}_{[-j]\lambda},0\mid \mathbf{y}_j)-
l_{inc}(\hat{\mu}_j,\hat{\eta}_j,0\mid {\mathbf y}_j)\nonumber\\
&&\qquad=\frac {n_j}2(\hat{\mu}_{[-j]\lambda}-\hat{\mu}_{j},\hat{\eta}_{[-j]\lambda}-\hat{\eta}_{j},\hat{\theta}_{[-j]\lambda}-\hat{\theta}_{j})\nonumber\\
&&\qquad\quad\times\frac 1{n_j}\frac{\partial^2l_{inc}(\hat{\mu}^*_{[-j]\lambda},\hat{\eta}^*_{[-j]\lambda},\hat{\theta}^*_{[-j]\lambda})}
{\partial (\mu,\eta,\theta)^T\partial (\mu,\eta,\theta)}
\begin{pmatrix}
\hat{\mu}_{[-j]\lambda}-\hat{\mu}_{j}\\
\hat{\eta}_{[-j]\lambda}-\hat{\eta}_{j}\\
\hat{\theta}_{[-j]\lambda}-\hat{\theta}_{j}
\end{pmatrix}\nonumber\\
&&\qquad=
\frac {n_j}2(\hat{\mu}_{[-j]\lambda}-\hat{\mu}_{j},\hat{\eta}_{[-j]\lambda}-\hat{\eta}_{j},\hat{\theta}_{[-j]\lambda})\nonumber\\
&&\qquad\quad\times
\begin{pmatrix}
-\frac 1{\sigma^2_0}&0&0\\
0&-2&0\\
0&0&0
\end{pmatrix}
\begin{pmatrix}
\hat{\mu}_{[-j]\lambda}-\hat{\mu}_{j}\\
\hat{\eta}_{[-j]\lambda}-\hat{\eta}_{j}\\
\hat{\theta}_{[-j]\lambda}
\end{pmatrix} (1+o_p(1))\nonumber\\
&&\qquad=
-\frac {n_j}2(\hat{\mu}_{[-j]\lambda}-\hat{\mu}_j,\hat{\eta}_{[-j]\lambda}-\hat{\eta}_j )I_{110}(\hat{\mu}_{[-j]\lambda}-\hat{\mu}_j,\hat{\eta}_{[-j]\lambda}-\hat{\eta}_j )^T\nonumber\\
&&\qquad\quad\times (1+o_p(1))\nonumber\\
&&\qquad=-\frac {n_j}2\left(\frac 1{n_{[-j]}}\sum_{i\in [-j]} \bu_i
-\frac 1{n_{j}}\sum_{i\in [j]}\bu_i\right)^T(-I_{110})^{-1}\nonumber\\
&&\qquad\quad\times
\left(\frac 1{n_{[-j]}}\sum_{i\in [-j]} \bu_i
-\frac 1{n_{j}}\sum_{i\in [j]}\bu_i\right) (1+o_p(1)).
\end{eqnarray}
We have
\begin{eqnarray*}
&&\cv(\lambda)+\frac 1K\sum_{j=1}^K
l_{inc}(\hat{\mu}_j,\hat{\eta}_j,0\mid {\mathbf y}_j)\\
&&\qquad=\frac 1{2K}\sum_{j=1}^Kn_j\left(\frac 1{n_{[-j]}}\sum_{i\in [-j]} \bu_i
-\frac 1{n_{j}}\sum_{i\in [j]}\bu_i\right)^T(-I_{110})^{-1}\\
&&\qquad\quad\times
\left(\frac 1{n_{[-j]}}\sum_{i\in [-j]} \bu_i
-\frac 1{n_{j}}\sum_{i\in [j]}\bu_i\right) (1+o_p(1)).
\end{eqnarray*}
The proof is completed.
\end{proof}
\end{appendix}


\begin{supplement}
\stitle{Activation functions}
\sdescription{Plots of hyperbolic transformation functions and their derivatives.}
\end{supplement}
\begin{supplement}
\stitle{Additional simulation results}
\sdescription{Box plots of estimation errors.}
\end{supplement}




\end{document}